\documentclass[10pt, conference, letterpaper]{IEEEtran}
\usepackage{booktabs} 
\usepackage[boxed]{algorithm2e}  
\usepackage{multirow,makecell}
\usepackage{setspace}
\usepackage{enumerate}
  \usepackage{mathrsfs}
\usepackage{graphicx}
\usepackage{subfigure}
\usepackage{amsthm}
\usepackage{amsmath,amssymb}
\usepackage{url}
\usepackage{color}
\usepackage{afterpage}
\usepackage{hyperref}
\usepackage{threeparttable}
\allowdisplaybreaks[2,4]
\newtheorem{definition}{Definition} 
\newtheorem{lemma}{Lemma} 
\newtheorem{theorem}{Theorem} 
\newtheorem{corollary}{Corollary}
\ifCLASSINFOpdf
  \else

\fi

\begin{document}
%
\title{Evolving Influence Maximization in Evolving Networks}


\author{\IEEEauthorblockN{Xudong Wu, Luoyi Fu, Zixin Zhang, Jingfan Meng, Xinbing Wang, and Guihai Chen}
\IEEEauthorblockA{Shanghai Jiao Tong University}
 \{xudongwu, yiluofu, zixin98, JeffMeng, xwang8\}@sjtu.edu.cn, gchen@cs.sjtu.edu.cn}

\maketitle
%



\begin{abstract}
Influence Maximization (IM) aims to maximize the number of people that become aware of a product by finding the `best' set of `seed' users to initiate the product advertisement. Unlike prior arts on static social networks containing fixed number of users, we undertake the first study of IM in more realistic evolving networks with temporally growing topology. The task of evolving IM ({\bfseries EIM}), however, is far more challenging over static cases  in the sense that seed selection should consider its impact on future users and the probabilities that users influence one another also evolve over time.

We address the challenges through $\mathbb{EIM}$, a newly proposed bandit-based framework that alternates between seed nodes selection and knowledge (i.e., nodes' growing speed and evolving influences) learning during network evolution. Remarkably, $\mathbb{EIM}$ involves three novel components to handle the uncertainties brought by evolution:
(1) A fully adaptive particle learning of nodes' growing speed for accurately estimating future influenced size, with real growing behaviors delineated by a set of weighted particles.
(2) A bandit-based refining method with growing arms to cope with the
 evolving influences via growing edges from previous influence diffusion feedbacks.
(3) {\bf Evo-IMM}, a priority based seed selection
algorithm with the objective to maximize the influence spread to highly attractive users during evolution. Theoretically, $\mathbb{EIM}$ returns a regret bound that provably
maintains its sublinearity with respect to the growing network size. Empirically, the effectiveness
of $\mathbb{EIM}$ are also validated, with three notable
million-scale evolving network datasets possessing complete social
relationships and nodes' joining time. The results confirm the superiority of $\mathbb{EIM}$ in terms of an up to $50\%$ larger influenced size over four static baselines.

\end{abstract}

\IEEEpeerreviewmaketitle
   \section {Introduction}\label{intro}
With the development of massive social networks (e.g., Facebook, Wechat and Twitter, etc), Influence Maximization (IM)  has become a key technology of viral marketing in modern business \cite{ICDE17, mobihoc, IMM}. Given a social network $G$ and an integer $K$, the goal of IM is to select $K$ seed users in $G$ in hope that their adoptions of a promoted product or idea can maximize the expected number of final adopted users through word-of-mouth effect \cite{icdm, retweet, quinn2015directed}. Initially put forwarded by Kempe et al. \cite{IM1}, the problem of IM has been intensively studied by a plethora of subsequent works,
proposing improvements or modifications from multiple aspects, including influence size estimation \cite{IMM, sigmetircs1, fIMM, cikm1}, adaptive seeding \cite{mobihoc, retweet}, boosting seeding \cite{ICDE17}, and many others.

The fundamental task in IM, as we noted above, lies in estimating the expected influenced size
of each alternative seed set based on each user's activation probabilities, referring to the the probability that a user successfully
influences his social neighbors after having been influenced himself. And
the \emph{influences} among users are quantified by those activation probabilities. While existing literature works well in finding the most influential seed users, they are all constrained to the assumption that the number of nodes in the network, along with their edges in between, are fixed during influence diffusion. Consequently, it violates real practices as many realistic social networks are usually growing over time. Take Wechat
\cite{Wechat}, the most popular social media app in China as an example.
The number of Wechat accounts (nodes) grew from zero to 300 million during its
early two years, with $410$ thousand new users per day on average, and are continuing to fastly approach almost 1
billion ones \cite{sigmoids}. And Facebook also exhibits a fast growth with roughly $340 K$ new users per day \cite{li2017time}. Similar phenomena also hold in a wide range of other real
social applications including Twitter, Academic networks,
etc. Meanwhile, a viral marketing action such as the web advertisements via messages or emails propagation may consume up weeks to months \cite{iribarren2009impact}. Thus, given an evolving network $G_{t}$ at time $t$ and time span $T$ for a viral marketing action, $G_{t}$ have greatly evolved to $G_{t+T}$ during influence diffused from seed users to the expected maximal size. Consequently, the expected influenced size estimated by existing IM techniques over $G_{t}$ cannot reflect the influence of seed set over $G_{t+T}$, which severely impacts the quality of selected seed users.

{\bf The above issue motivates the study of evolving influence maximization ({\bf EIM}), whose problem formulation should incorporate the evolutionary nature of $G$ during propagation.} \emph{Interpreted technically, given an instance of evolving social network $G_t$ at time $t$ and an integer
$K$, the goal of  {\bf EIM} is to select $K$ seed users to maximize the influence diffused to both existing users and those will join during time $t$ to $t+T$.}
 Different from the well investigated existing IM problems,
the task of {\bf EIM} turns out to be highly non-trivial due to the following three
challenges in reality:
(1) The growing speed of a specific network exhibits uncertainties due to multiple external factors (e.g., the number of potential users, user interests and peer competitions). Such uncertain growing speed hinders accurately predicting how the network evolve during time $t$ to $t+T$.
(2) There is no prior knowledge about the influences via newly emerged edges, and they may also evolve over time with the changes of social relations among users (e.g., from friends to strangers or on the contrast). Although some recent efforts \cite{icdm, retweet, cosn, semi-bandit, iwqos} have been dedicated to online IM where influences among users are uncertain,
the underlying network topology is still assumed to be completely known, thus inapplicable to the situations with both growing nodes and edges. (3) In evolving networks, newly added users are more inclined to establish relationship with those of higher popularity. Thus users in $G_{t}$ have different attractiveness to new users, as opposed to existing IM studies which treat each user equally. Unfortunately, as far as we know, no studies have been directed toward IM in temporally growing networks. Consequently, it remains open how to effectively resolve {\bf EIM} that can jointly deal with the unknown influences, uncertain growing speed and heterogeneous attractiveness.


\textbf{This motivates us to present a first look into {\bf EIM} problem}. By proving its NP-hardness, we attempt to solve the above three challenges in {\bf EIM} by $\mathbb{EIM}$, a new and novel bandit-based {\bfseries E}volving {\bfseries I}nfluence {\bfseries M}aximization framework with multiple \emph{periods} of IM campaigns \footnote{\vspace{-1mm}We shall elaborate in Section 3 the reason for choosing the bandit-based framework and the incorporation of multiple periods.}. Each period amounts to an IM campaign which chooses seeds that improve the knowledge and/or that lead to a large spread to both existing users and those that will join till the end of this period, and incurs a \emph{regret} in such influence spread due to the lack of network knowledge. Different from prior IM studies, here the network knowledge includes the networks' growing speed and evolving influences via continuously emerging new edges in network evolution. Thus $\mathbb{EIM}$ seeks to minimize the accumulated regret incurred by choosing suboptimal seeds over multiple periods.

 While we defer the details of $\mathbb{EIM}$ design in later sections (Sections \ref{Modeling Evolving}, \ref{learning}, \ref{evolving influence maximization}), here we briefly unfold its three novel components in addressing the aforementioned three challenges in {\bf EIM}:

(1) It is unrealistic to assume the complete network topology is known in advance, thus a fully adaptive particle learning method is proposed to capture the uncertain network growing speed, with real growing function of nodes explicitly represented by a set of weighted particles. By modeling network evolution via the popular Preferential Attachment (PA) rule (i.e., new users prefer connecting to higher degree nodes), we are able to predict potential added users during influence diffusion with weighted particles (Section \ref{Modeling Evolving}).

(2) Considering the evolving influences among users, we model
the influences via continuously emerging edges as the growing arms
in the bandits, thus ensuring the applicability of $\mathbb{EIM}$ to the evolving network with growing nodes and edges (Section 5). By modeling the activating probabilities as the dynamic rewards distribution of the arms, the reward
of each edge as the edge-level feedback can then be taken to adaptively refine the
estimating values of the evolving influences.


(3) Aiming at maximizing the influence diffused to both existing and future users, we introduce a novel priority based seed selection algorithm {\bfseries Evo-IMM} that incorporates the heterogeneity of users' attractiveness formed by PA rule (Section \ref{evolving influence maximization}). In {\bfseries Evo-IMM}, users with higher attractiveness to future ones are sampled with higher priority in seeds selection. {\bfseries Evo-IMM} turns out to provably enjoy comparable performance such as approximation ratio and time complexity with the static counterparts.

 We validate the performance of $\mathbb{EIM}$ from both theoretical and empirical perspectives. Theoretically, although the growing size and successive emerging orders of the arms duo to network evolution further challenges the knowledge learning compared to classical bandits, the regret bound of $\mathbb{EIM}$ still provably maintains to be sublinear to the number of trials under the growing network size (Section \ref{Performance Analysis}). Empirically, the effectiveness of $\mathbb{ EIM}$ is validated on both synthetic and real world evolving networks, with up to 200 years of time span and million scale data size respectively (Section \ref{experiments}). Notably, the real evolving networks are extracted from the true academic networks with complete co-authorship, citation and joining time of all authors and papers, which is severely lacking in existing IM works.
Experimental results demonstrate the superiority of $\mathbb{EIM}$. For example, $\mathbb{EIM}$ achieves a $50\%$ lager influenced size than four static baselines in an evolving Co-author network with $1.7$ million nodes. 

  \section{Related Works}
 \vspace{-1mm}
\subsection{Static Influence Maximization Problem}\label{Preliminaries}
  \vspace{-1mm}
Kempe. et al. \cite{IM1} are the first to formulate influence maximization problem over a given network as a combinatorial optimization problem. Particularly, in their seminal work \cite{IM1}, they treat the network as a graph $G=(V, E)$, where there is an influence cascade process
 triggered by a small number of influenced users that are called \emph{seed} users. The influence diffusion process is then characterized by the later widely adopted Independent Cascade (IC) model \cite{ICDE17}-\cite{IM1}, whose definition is given as follows:
\begin{definition}\label{IC}
(Independent Cascade (IC) model.) In the IC model, the influences among users are characterized by the activation probabilities. Specifically, once user $u_{i}$ is influenced, he has a single chance to activate his social neighbor $u_{j}$ successfully with activation probability $p_{ij}$ via edge between users $u_{i}$ and $u_{j}$). And whether or not $u_{i}$ can influence $u_{j}$ successfully is independent of the history of information diffusion.
\end{definition}
For a given seed set $S$, let $I(S, G)$ be the expected number of users that are finally influenced by  the seed users in $S$ estimated under the IC model. The objective of IM is to find a set of $K$ seed users (i.e., $S^{opt}$) who can maximize $I(S,G)$ among all the sets of users with size $K$. That is,
\vspace{-1mm}
\begin{equation}\label{classical IM}
S^{opt}=\mathbf{\mathop{\arg\max}}_{S\subseteq V, |S|=K} \, I(S, G).
\vspace{-1mm}
\end{equation}

Based on the above formulation, Kempe. et al. prove the NP-hardness of the IM problem, and design the greedy algorithm that provably returns a $(1-1/e)$-approximate solution for seed selection. Since then, a large number of subsequent works have emerged to improve the efficiency and quality of IM designing. For some representative examples, \cite{IMM}, \cite{fIMM} and \cite{cikm1} focus on achieving reasonable complexity in seed selection over million or even billion-scale networks.
Besides, different costs for seeding different users are considered in \cite{mobihoc} and \cite{yang2016continuous} for the cost-aware IM problems, with the corresponding near optimal budget allocation methods proposed.

The objectives of the above works are all set to select the seed set with the maximum $I(S, G)$ estimated over the static network $G$. As a result, over evolving networks where new users continuously join in and influences evolve over time, it is difficult for classical IM techniques to return high quality seeds since $I(S, G)$ estimated by them fails to include the future users and their influences.

 \subsection{Dynamic Influence Maximization Problem}
 As a step ahead of classical IM problems, some recent attempts are made in dynamic networks. For example, considering the network with dynamically changing edges, \cite{cosn} takes multiple specific examples to show the effect of changing typologies on IM design, and highlights the importance of seeding time. Similarly,  the effect of dynamic user availability is studied in \cite{jankowski2013compensatory}, and the effect of seeding time are also experimentally shown.  To cope with the unknown influences among users, Quinn \emph{et al.} \cite{quinn2015directed} proposed to learn the influences from previous information propagation activities. Although serval effective algorithms are designed in \cite{quinn2015directed} for learning uncertain influences, they are merely applicable to the network with static users. Besides,  Michalski \emph{et al.} \cite{michalski2014seed} focus on maximizing the influence diffused to multiple given network snapshots. However, they assume that future network is known in advance, which violates the real practices. 

 Meanwhile, there emerges a class of online IM techniques that periodically seed one or more users in dynamic networks, in a similar manner to our settings that will be described later. To unfold, Tong \emph{et al.}                                
\cite{tong2017adaptive}  propose to successive select seed users with influence diffusion over dynamic networks, while just considers changing edges among fixed users. Besides, considering the Multi-Arm Bandits (MAB) is a widely used framework that learns dynamics and make reasonable decision as possible \cite{liu2013learning}, the bandit-based learning framework is adopted in \cite{semi-bandit} and \cite{retweet} to refine unknown influences from the feedbacks of previous influence diffusion, and periodically select a set of seed users under the refined influences.
Regardless of their progress, those online IM still considers the uncertain influences over static network topology,
where the estimated $I(S, G)$ also fails to include the influence diffused to the future users.
Thus it is still difficult for the seeds to be repeatedly selected at different time to meet requirement of high quality.

 As far as we know, the only work that shares the closest correlation with us
belongs to Li et. al. \cite{li2017time}, who simulate the network growth based on the Forest Fire Model and then run the existing static IM algorithms over the simulated network. However, under the unknown growing speed, it is difficult for the simulation to capture the real network evolution. Furthermore, influences among users are still preset as known constants.
The limitations of the state-of-art IM techniques motivates us to study evolving influence maximization, which will be formally defined in next section.

\section{Evolving Influence Maximization}\label{model}

  \vspace{-1mm}

\subsection{Problem Formulation}\label{EIMproblem}

{\bfseries  Evolving IM problem (EIM).} We assume that time is divided into different time stamps.  And an evolving network at time stamp $t$ is modeled as a graph $G_{t}=(V_{t}, E_{t})$, where $V_{t}$ and $E_{t}$ respectively denote users and their relationships in $G_{t}$. Given an IM campaign that takes $T$ time (which is called as survival time later), the network may evolve from $G_{t}$ to $G_{T+t}$ during influence diffusion with newly added nodes and edges. Thus, different from the classical IM problem defined in Eqn. (\ref{classical IM}), we redefine the evolving IM problem over $G_{t}=(V_{t}, E_{t})$  as follows.

\begin{definition}\label{problem statement}
({\bf EIM} problem.) Given an evolving network at timestamp $t$, i.e., $G_{t}=(V_{t}, E_{t})$ and the survival time $T$ of an IM campaign, the objective of {\bf EIM} is to find a set of users $S^{opt}$ with size $K$ to maximize the influence spread to both users in $V_{t}$ and those that will join during $t$ to $t+T$. That is, we aim at solving
\begin{equation}\label{EIM problem}
S^{opt}=\mathbf{\mathop{\arg\max}}_{S\subseteq V_{t}, |S|=K} \, I(S, G_{t+T}).
\end{equation}
\end{definition}
Note that in Definition \ref{problem statement}, the seeds are selected from the current network $G_{t}$ instead of the future instances $G_{t'} ( t<t'\leq T)$. The reason behind is that the existing network $G_{t}$ is known, while it is difficlt to know which users will be in the future network instances and how they will be connected to each other. Since assuming the future instances $G_{t'} ( t<t'\leq T)$ known at time $t$  is unrealistic, it is more reasonable to select the seed set from the current $G_{t}$, with the objective being maximizing the influence diffused over $G_{t+T}$. Similar to the classical IM problem, the {\bfseries EIM} defined above is also NP-hard. Lemma \ref{hardness} states the hardness of {\bfseries EIM} problem and the submodularity of its objective function $I(S, G_{t+T})$.
\vspace{-2mm}
\begin{lemma}\label{hardness}
The {\bfseries EIM} problem is {\emph NP-hard}. The computation of $I(S, G_{t+T})$ is {\emph \#P-hard}. And the objective function $I(S, G_{t+T})$ is monotone and submodular\footnote{A set function $I(\cdot)$ is  monotone if $I(A)\leq I(B)$ for all $A \subseteq B$, and $I(\cdot)$ is submodular if $I(A\cup x)-I(A)\geq I(B\cup x)-I(B)$ for all $A \subseteq B$.}.
\vspace{-2mm}
\end{lemma}
\vspace{-2mm}
\begin{proof}
The NP-hardness and \#P-hardness can be respectively proved by the reductions of NP-completed \emph{Set Cover} problem and \#P-completed \emph{S-D connectivities} counting problem. And the submodularity of $I(S, G_{t+T})$ can be proved  by modeling the additional influence brought by a new seed as the marginal gain from adding an element to the set $S$. We leave the detailed analysis in the Appendix \ref{hardnessproof1}.
\end{proof}
\vspace{-2mm}
{\bf Challenges of solving \textbf{EIM}.} The NP-hardness of {\bfseries EIM} implies the necessity to seek for approximate algorithms for seed selection. However, as noted in Section 1, solving {\bf EIM} is far more challenging due to the evolving nature of the network included. Under Definition 3.1, the three challenges can be reproduced as: (1) The unknown growing speed makes it difficult to predict how many new users in $V_{t+T}$ will connect to existing users in $V_{t}$; (2) The influences among users evolve over time, which, together with the unknown growing speed, renders it impossible to accurately estimate $I(S,G_{t+T})$. (3) The heterogeneous attractiveness infers that users in $V_{t}$ cannot be equally treated in seed selection.

\subsection{Overview of $\mathbb{ EIM}$}
\vspace{-1mm}
Regarding the three challenges above, we propose a new framework that can better incorporate the evolving nature in solving \textbf{EIM}. We note that what is built upon the
three challenges, as also indicated in Section 1, is that the survival
time of an IM campaign only varies from weeks to months in
reality, leading to users joining the network several months later
unable to be influenced by this early IM campaign. Consequently, only selecting the seed users in the beginning and 
triggering an IM campaign once under
the uncertain network knowledge will severely restrict the long
term profits obtained from viral marketing.

\subsubsection{Basic idea of solving $\mathbb{ EIM}$}
We thus try to maximize the influence diffusion size over such evolving network by solving {\bfseries EIM} in \emph{multiple periods}, with one period corresponding to the survival time $T$ of an IM campaign and a set of new users seeded at the beginning of each period. Given that the initial network is $G_{t}$ and $T$, the objective of {\bf EIM} in the first period is to select a set $S$ of seeds from $V_{t}$ to maximize $I(S, G_{t+T})$ defined in Definition \ref{problem statement}. And the objective in the second period is to select a set $S$ from $V_{T+t}$ to maximize $I(S, G_{t+2T})$. Similar manner holds in subsequent periods. Thus successive IM campaigns in multiple periods give chance to maximize the number of influenced users in a long term. Meanwhile, the periodical seed selection also enables us to cope with the three challenges. To elaborate, users join during pervious periods are the natural samples to learn the growing speed at a given period. And the evolving influences among users can be learnt from the activating results during previous influence diffusion. Therefore, to systematically resolve the above three challenges,
each period consists of the following three steps:
(1) Learning network growing speed from the feedbacks of observed newly added users.
(2) Learning evolving influences from previous influence diffusion feedbacks.
%
(3) Selecting seed set for triggering an IM campaign under the refined network knowledge in above two steps. Taking timestamp $t$ as an example, the objective of step (1) is to predict the network structure until time $(t+T)$, and  step (2) aims at obtaining real influences among users to accurately estimate $I(S, G_{t+T})$ for any seeds set $S$. Then step (3) focus on selecting the seeds set $S$ who can maximize $I(S, G_{t+T})$. 
 With the number of total periods being set as $R$, all the IM applications in diverse scenarios can be well characterized by simply adjusting the values of $R$ and $T$. The objective of our solution is equivalent to maximizing the sum of influenced size during the $R$ periods. 

While we unfold the details of the three steps in Sections \ref{Modeling Evolving}, \ref{learning} and \ref{evolving influence maximization} respectively,  we remark that the idea of periodical seed selection in EIM cannot be trivially extended from that in recent online IM studies. As pointed out in Section 2.2,  it is because the dynamic influences are restricted among fixed number of users in online IM, while seeds in {\bf EIM} are selected from continuously joining users and the objective is to maximize the influence diffused to both the existing and future users. With this regard, existing online IM can be reduced as a special case of {\bf EIM} by simply letting users in the network remain static over time.
\begin{figure}[t]
 \centering
 \vspace{-1mm}
\centering
  \includegraphics[width=0.48\textwidth]{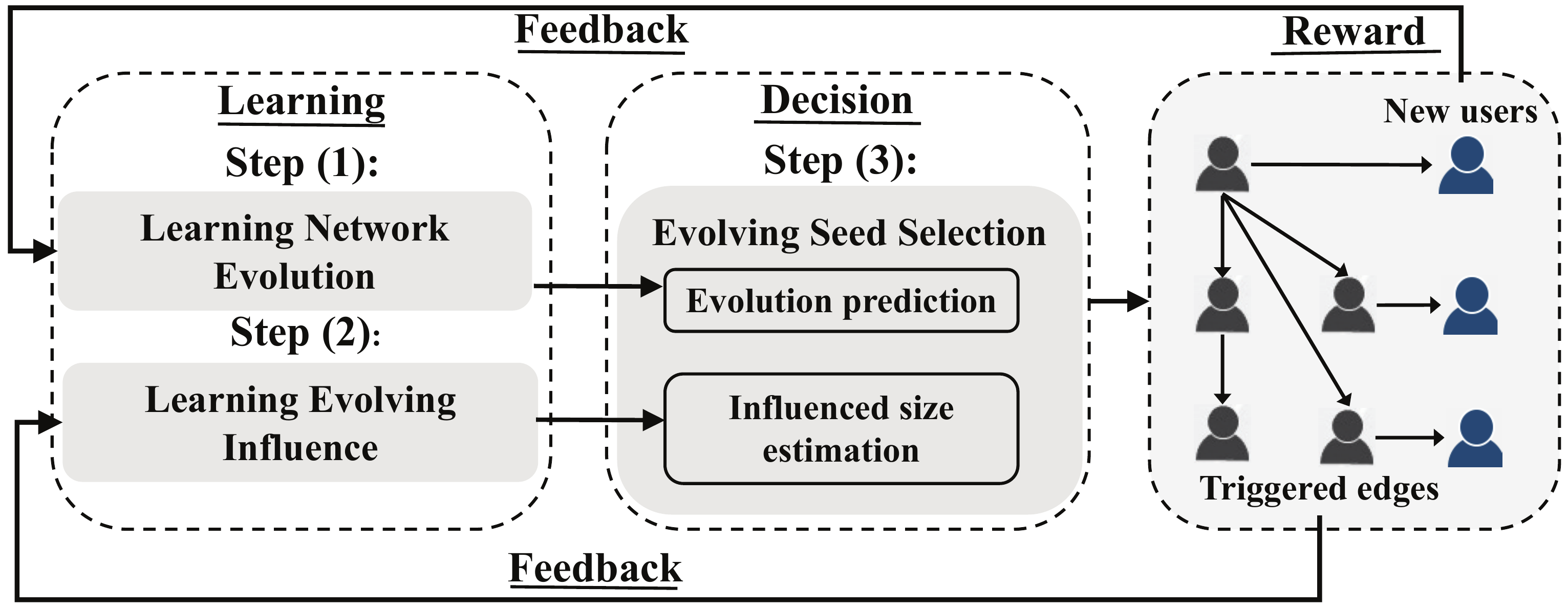}
  \vspace{-3mm}
  \caption{Overview of $\mathbb{EIM}$ in $r$-th trial}\label{EvoIMMoverview}
   \vspace{-6mm}
  \end{figure}

\subsubsection{Adaption to Combinatorial Multi Arm Bandits (CMAB)} Note that the above three steps in each period naturally forms a learning-decision process, where we first learn the growing speed and evolving influences from previous period and then decide which users to seed.  To this end, we design a novel framework $\mathbb{EIM}$ to coordinate the above three steps in multiple periods, as illustrated in Figure \ref{EvoIMMoverview}. $\mathbb{EIM}$ allows to convert the {\bf EIM} problem into a Combinatorial Multi Arm Bandits (CMAB) one reviewed below:

 In general CMAB,  there are $m$ arms with unknown reward distributions and, in each trial, it makes a decision that chooses a set of arms with maximum expected rewards to trigger \cite{liu2013learning}. Then the reward obtained from each arm is taken as the feedback to update its reward distribution, and in next trial, the decision is made under the updated rewards distribution. Given the total number of trials  $R$, the objective of CMAB is to design an arm selection strategy to maximize the long term rewards obtained from the trials.

 Regarding this, in $\mathbb{EIM}$ we model an IM campaign as one trial and totally $R$ trials will be performed. The decision in the $r$-th ($0<r\leq R$) trial is to select the seeds from the evolving network $G^{r}=(V^{r}, E^{r})$ at time \footnote{Throughout the rest of the paper, we have the following relations: The initial network is $G(t)$ at time $t$, where we set $T^1=t$.  Let $T^{r}$ represent the time when the $r$-th trial occurs. The evolved network at $T^{r}$ is denoted as $G^r=(V^r, E^r)$, with $V^r$ and $E^r$ being the corresponding evolved node and edge sets at $T^r$.} $T^{r}$ . The triggered arms correspond to the activated edges in the influence diffusion starting from the selected seed nodes under IC model during $T^{r}$ to $T^{r+1}$. By modeling the activation probabilities as the reward distributions, we consider the edge-level feedback in $\mathbb{EIM}$ where we can observe wether the activation via an edge is successful or not.
 
Table \ref{mapping} lists the
mapping of the various components of CMAB to $\mathbb{EIM}$ framework. Different from the general CMAB, here the number of arms in $\mathbb{EIM}$ grows with the continuously emerging new edges during network evolution.

\begin{table}[h]
\vspace{-3mm}
\centering
\caption{Mappings between CMAB and EIM}\label{mapping}
    \vspace{-3mm}
{\footnotesize
\begin{tabular}{c|c|l} \hline
{\bfseries CMAB} &{\bfseries Symbol} &{\bfseries  EIM}
\\ \hline
$r$-th trial& r &   IM campaign in $r$-th period \\ \hline
Arm &$e$&Influence via edge $e$  \\ \hline
Reward of arm $e$ &$z_{e}$&Activating result of edge $e$ \\ \hline
\multirowcell{2}{Reward of $r$-th trial}  &$I(S^{r}, G^{r+1})$& Influenced size during \\
   &&$T^{r}$ to $T^{r+1}$  \\ \hline
\multirowcell{3}{Bandits feedback}&$\Delta n(T^{r})$ &Observed new users during  \\
&&$T^{r}$ to $T^{r+1}$\\
                                                 &$z_{e}$ &Reward of edge $e$\\
\hline
\end{tabular}}
\vspace{-1mm}
\end{table}

{\bfseries Example.} We further give an example to facilitate the understanding of $\mathbb{EIM}$. Let the budget for a viral marketing be seeding $60$ users and the survival time for an IM campaign be one month. $\mathbb{EIM}$ divides the viral marketing into multiple trials by seeding $5$ users one month. Here, two consecutive trials are one month apart. Suppose that the initial network starts at 1st, May, and the objective of the first trial is to select $5$ seeds from current users to maximize the influence among those joining before 1st, May and during 1st, May to 31st, May. Then the second trial is on 1st June with the  corresponding objective being maximizing influences among users joining until 30th, June, etc. In the $r$-th trial, $\mathbb{EIM}$ first learns the network knowledge from influence diffusion feedbacks during previous $(r-1)$ months, and then selects $5$ users to maximize the influence during the $r$-th month based on the refined growing speed and evolving influences. The reward of $\mathbb{EIM}$ in this example is the influenced size during the $12$ months.

{\bfseries Remark.} In the present work, we focus on the case where the network exhibits fast growth while the promoted information remains effective in a far longer period. However, we do not need to rely on  any correlation between the speed of newly added users and that of influence propagation. As long as the network is evolving, $\mathbb{EIM}$ can adaptively capture its growing speed, and then selects seed users under the learnt growing speed in each trial. Even if the network is static, $\mathbb{EIM}$ is also applicable by setting $G_{t+T}=G_{t}$.

\section{Learning Network Evolution}\label{Modeling Evolving}
In this section, we dive into the first step, i.e, learning the future netowrk evolution during influence diffusion in the proposed $\mathbb{EIM}$ framework. To unfold, we need to address the following two questions: (1) How the newly added users connect with existing users; (2) How many new users will join in during influence diffusion.

\subsection{ Preferential Attachment (PA) Rule}

For the first question, we adopt the well-known Barab$\acute a$si-Albert (BA) model \cite{BA1, BA2} to characterize the evolution of social networks. BA model is capable of well capturing the typical features, i.e., power-law degree distribution, shrinking diameter and clustering structure that exist in most real social networks.
The evolution under BA model is interpreted as follows: a new node joins the network at each time slot $\Delta t$, and establishes $m$ new edges with the existing nodes ($m$ is a constant) \cite{sigmoids, BA2}.
Let $V_{t}$ denote the set of users at time $t$, and $d_{n}^{t}$ denote the current degree of node $v_{n}\in V_{t}$. For a newly added user at time $t$, it establishes a new edge with a chosen existing user $v_{s}$ in each time slot $\Delta t$ according to the rule of Preferential Attachment (PA), meaning that  the probability of choosing $v_{s}$ is proportional to its current degree. Then the remaining $(m-1)$ edges are respectively established in next $(m-1)$ time slots in the same manner. 

{\bf Remark.} Although $m$ is set as a constant in the BA model \cite{BA1, BA2}, it can still capture the evolution of most networks, with the statistical property of real social networks being that each newly added node expectedly establishes a same number of new edges \cite{sigmoids}. The BA model will also be empirically justified in Section \ref{datasets} under various real datasets, all of which exhibit the phenomenon of ``Richer gets richer".

Under the PA rule, the expected degree of node $v_{n}$ at time slot $t+\Delta t$ is equal to

\vspace{-3mm}
 \begin{equation}
\label{PR}\mathbb{E}(d_{n}^{t+\Delta t})=d_{n}^{t}\cdot \left(1+\frac{1}{\sum_{v_{j}\in V_{t}}d_{j}^{t}+1}\right).
\vspace{-1mm}
\end{equation}

Given the number of users in evolution at time $t$ is $n(t)$, the period $T$ of each trial in $\mathbb{EIM}$ is consisted of $m[n(t+T)-n(t)]$ evolving slots since there are $[n(t+T)-n(t)]$ newly added users and each user brings $m$ new edges during the time span $T$. Based on the PA rule, Lemma \ref{degree} gives the expected degree of a given node in evolution.

\begin{lemma}\label{degree}
Given the degree of node $v_{n}$ at time $t$ is $d_{n}^{t}$ and the period $T$ of each trial,  we have
\begin{equation}
\notag \mathbb{E}(d_{n}^{T+t})=d_{n}^{t}\cdot \prod_{s=1}^{m[n(t+T)-n(t)]} \left(1+\frac{1}{\sum_{v_{j}\in V_{t}}d_{j}^{t}+(2s-1)}\right).
\end{equation}
\end{lemma}

The proof for Lemma \ref{degree} is shown in Appendix \ref{Lemma 4.1}.

 Under PA rule, Lemma \ref{degree} returns the expected degrees of the existing users determined by the given growing speed $n(t)$. However, how to determine the network growing speed function $n(t)$ for a given specific network? This is the second question to be answered in this section, and will be addressed in the following.
\subsection{Learning Networks' Growing Speed}\label{Learning nodes growing speed}
Now we proceed to illustrate the network growing speed learning method in $\mathbb{EIM}$, which answers the second question posted at the beginning of this section. Note that $\mathbb{EIM}$ is a bandit-based framework, and our method for growing speed learning utilizes the bandits feedback. Before we give the learning method for the growing speed $n(t)$, we need to understand how the real network grows. In reality, as noted in Section \ref{intro}, the network growing speed $n(t)$ is affected by multiple factors.
To elaborate, at time $t$, the $n(t)$ existing users prefer to attract new users to join the network, while the total population $N$ of who can join is limited \cite{sigmoids}. As a result, the growing speed is constrained by the term $[N-n(t)]$. On the other hand, users exhibit decaying interests  $\frac{\beta}{t^{\theta}}$ in attracting users to join \cite{sigmoids}, in a similar manner to the susceptible infected (SI) model in epidemiology \cite{SI}. The exponent $\theta$ reflects the growing dynamics such as power law, linear, sub-linear, etc. Jointly considering the above factors, we adopt the Nettide-node model \cite{sigmoids} to characterize the networks' growing speed, which is expressed as
\vspace{-1mm}
\begin{equation}\label{speed}
\frac{dn(t)}{dt}=\frac{\beta}{t^{\theta}}n(t)[N-n(t)].
\vspace{-1mm}
\end{equation}
The Nettide-node model has been previously empirically justified over real social network data (e.g., Facebook, Wechat, Google-plus and arXiv, etc) \cite{sigmoids} in terms of its effectiveness in capturing networks' growing speed, with an error of less than $3\%$. However, under the assumption of unknown future network topology, the parameters (i.e., $\beta, \theta$ and $N$) of a specific evolving network are unknown in advance. Thus determining the network growing speed becomes learning the three parameters $\beta, \theta$ and $N$ in Eqn. (\ref{speed}).  To this end, we propose a fully adaptive particle learning method to adaptively capture the nodes growing speed. In the particle learning method, we use each particle to represent a possible network growth speed and 
the definition of particles is given below.
\vspace{-1mm}
\begin{definition}\label{particle}
(Particle.) Each particle $\rho_{i}$ represents a growing speed function with given prior parameters ($\beta_{i}$, $\theta_{i}$ and $N_{i}$), i.e., $\frac{dn_{i}(t)}{dt}=\frac{\beta_{i}}{t^{\theta_{i}}}\cdot n_{i}(t)[N_{i}-n_{i}(t)]$. And they will be resampled based on their weights $w_{i}$ in each trial.
\vspace{-1mm}
\end{definition}

Given the definition, before we show the learning process, we briefly introduce the main idea of particle learning. We first take each possible growing speed function into Lemma \ref{degree} to predict the future degrees of existing nodes, which serve as the prior value of the corresponding particle. With the observed real degrees serving as the posterior value, the difference between a particle's prior value and the posterior value is used to determine its weight, which quantifies its reliability in reflecting real growing speed. Based the above main idea, we now move to the elaborate learning process that relies on the general resample-propagate process, which is considered as an optimal and fully adaptive framework in particle learning \cite{particle}.
In correspondence to the {\bf EIM} problem, the resampling and propagation phase respectively refer to the growing speed refining and evolution prediction described below.


\underline{Growing speed learning:} In the $1$-st trial, the particle learning is initialized by a set of particles $\mathcal{P}^{1}$ with randomly sampled prior parameters ($\beta$, $\theta$ and $N$) from their possible ranges, which will also be empirically presented in Section \ref{performance of particle}. With the progress of $\mathbb{EIM}$, the simulated evolving process under each particle are proceeded in parallel. Specifically, at the beginning the of the $r$-th trial (i.e., the timestamp at $T^{r}$), we first take the growing function $n_{i}(t)$ into Lemma \ref{degree} to compute the expected degrees of current nodes until the end of the $r$-th trial (i.e., the timestamp at $T^{r+1}$). We use $\mathbb{E}_{i}(d_{e}^{r+1})$ to denote the expected degree of node $v_{e}$ until time $T^{r+1}$ under the condition that the growing speed is $n_{i}(t)$, and $\mathbb{E}_{i}(d_{e}^{r+1})$ serves as the prior value of particle $\rho_{i}$. The detailed derivations for $\mathbb{E}_{i}(d_{e}^{r+1})$ is deferred to Appendix \ref{derivations for added degrees}. Upon the influence diffusion ended at time $T^{r+1}$, the real degrees of the influenced nodes are counted to compute the posterior value.  Consider the fact that social medias (e.g., Twitter and Weibo) can track the activities of their users such as one user retweeting a tweet forwarded by another user \cite{retweet}\cite{semi-bandit}, in influence diffusion, the neighbors of influenced users and only the neighbors of influenced users can be observed. Certainly, the influenced users as well as their current degrees are also observable. Let $O(T^{r})$ denote the set of nodes that are influenced in the $r$-th trial, and  $O(T^{r})\cap\big(\cup_{i=1}^{r-1}O(T^{i})\big)$ denote those that are influenced not only in the $r$-th trial but also in one or more of the previous $(r-1)$ trials. For each node $v_{e} \in O(T^{r})\cap\big(\cup_{i=1}^{r-1}O(T^{i})\big)$, given its last observed time being $T^{(e, 0)}$ and the corresponding degree being  $d_{e}^{(e, 0)}$, the prior value of particle $\rho_{i}$ is equal to 
\begin{equation}\label{added nodes}
\Delta n_{i}(T^{r+1})=\sum_{O(T^{r})\cap\big(\cup_{i=1}^{r-1}O(T^{i})\big)}\left (\mathbb{E}_{i}(d_{e}^{r+1})-d_{e}^{(e,0)}\right), 
\end{equation}which is the sum of the expected incremental degrees of nodes in $O(T^{r})\cap\big(\cup_{i=1}^{r-1}O(T^{i})\big)$. On the other hand, when $v_{e}$ is influenced in the $r$-th trial, its real degree at time $T^{r}$ is observed, and we denote it by $d_{e}^{r+1}$. Thus the real degrees of nodes in  $O(T^{r})\cap\big(\cup_{i=1}^{r-1}O(T^{i})\big)$ can be taken as the ground truth in particle learning, and the posterior value of the particles is  determined as 
\begin{equation}\label{added nodes1}
\Delta n(T^{r})=\sum_{O(T^{r})\cap\big(\cup_{i=1}^{r-1}O(T^{i})\big)}\left (d_{e}^{r+1}-d^{(e,0)} \right). 
\end{equation}Under the prior value  $\Delta n_{i}(T^{r})$ and the posterior value $\Delta n(T^{r})$  which are respectively determined in Eqn. (\ref{added nodes}) and  Eqn. (\ref{added nodes1}), the weight of particle $\rho_{i}$ is inversely proportional to the square error between $\Delta n(T^{r})$ and $\Delta n_{i}(T^{r})$. That is, 
\begin{equation}\label{particle weight}
w_{i}(T^{r})\propto 1/ (\Delta n(T^{r})-\Delta n_{i}(T^{r}))^{2}.
\end{equation}

Based on the weights of particles, a resampling process is conducted to resample particles set $\mathcal{P}^{r}$ from those in $\mathcal{P}^{r-1}$  with the number proportional to their weights, and the total number always satisfies  $|\mathcal{P}^{r}|=M\, (0\leq r\leq R)$. The objective of resampling phase is to resample the particles whose growing functions near the ground truth as more new particles, and simultaneously kill those with large deviations from the ground truth. Following the resampling phase, the propagation phase, which corresponds to the evolution prediction in {\bf EIM} problem, is conducted to predict the real network evolution with the resampled particles.

\underline{Evolution prediction:} Following the resampling phase, we compute the expected incremental degrees of nodes in $V^{r+1}$ until time $T^{r+2}$ under each resampled particle (i.e., $\mathbb{E}_{i}(\Delta d_{e}^{r+2}$), $\rho_{i}\in \mathcal{P}^{r}$). 
And the expected incremental degree of $v_{e}$ from $T^{(e,0)}$ to $T^{r+1}$ can be computed as
\begin{equation}\label{added degrees2}
\mathbb{E}_{i}(\Delta d_{e}^{r+2})=\mathbb{E}_{i}(d_{e}^{r+2})-d^{(e,0)}.
\end{equation}Then we set the incremental degree of node $v_{e}$ as the average of expectation under each particle, i.e., $\mathbb{E}(\Delta d_{e}^{r+2})=\sum_{i=1}^{M}\frac{1}{M}\mathbb{E}_{i}(\Delta d_{e}^{r+2}$)), which represents the expected number of new neighbors of current nodes until time $T^{r+2}$ and is also the quantization of their attractiveness during $T^{r+1}$ to $T^{r+2}$. 
\SetAlCapSkip{0.2em}
\begin{algorithm}[htbp]
	\begin{small}
		\begin{spacing}{0.8}
			\SetAlgoLined
			
			// \textit{Particle learning in the $r$-th trial}\\
\KwIn{Influenced nodes set $O(T^{r})$, particles set: $\mathcal{P}^{r-1}$\;}
\KwOut{Particles set: $\mathcal{P}^{r}$, incremental degree: $\mathbb{E}(\Delta d_{e}^{r+1})(v_{e}\in V^{r}) $\;}  
// \textit{Resampling phase}\\
Count the number of newly added nodes: $\Delta n(T^{r})$;\\
\For{each $\rho_{i} \in \mathcal{P}^{r-1}$}{
\For{each $v_{e} \in O(T^{r})\cap\big(\cup_{i=1}^{r-1}O(T^{i})\big)$}{

Compute expected degree: $\mathbb{E}_{i}(d_{e}^{r})$\;
}
Compute prior value: $\Delta n_{i}(T^{r})$ (Eqn. (\ref{added nodes}))\;
Compute its weight: $w_{i}(T^{r})$ (Eqn. (\ref{particle weight}))\;

}
Resample particles with weights: $\mathcal{P}^{r-1}\rightarrow \mathcal{P}^{r}$\;
// \textit{Propagation phase}\\
\For{each $v_{e} \in V^{r+1}$}{
Compute incremental degree:  $\mathbb{E}(\Delta d_{e}^{r+2})=\sum_{i=1}^{M}\frac{1}{M}\mathbb{E}_{i}(\Delta d_{e}^{r+2})$\;}

\Return 
$\mathbb{E}(\Delta d_{e}^{r+2}) (v_{e}\in V^{r+1})$ and $\mathcal{P}^{r}$.
			\end{spacing}
		\end{small}
		\caption{{ Learning network evolution ({\bf Evo-NE)}.}}\label{PLA}
	\end{algorithm}

The pseudo code of the above particle learning process, which is mainly composed of resampling-propagation phases, is further summarized in Algorithm \ref{PLA} called {\bf Evo-NE}. Algorithm \ref{PLA}  takes the influenced nodes during $T^{r-1}$ to $T^{r}$ and particles set $\mathcal{P}^{r-1}$ as the input, and the prior value of each particle is computed as Eqn. (\ref{added nodes}). Then the particles in $\mathcal{P}^{r-1}$ are resampled as new particles $\mathcal{P}^{r}$ based on their weights determined by Eqn. (\ref{particle weight}). Following the resampling phase, we compute the expected incremental degrees of nodes in $V^{r}$ with the resampled phase in $\mathcal{P}^{r}$ representing the predicted network evolution during $T^{r}$ to $T^{r+1}$. The complexity of Algorithm \ref{PLA} is shown as below. 
			
{\bfseries Complexity.} In resampling phase,  {\bf Evo-NE} needs to traverse all the nodes under each particle in $\mathcal{P}^{r-1}$ to compute the prior value and weight of each particle. Then in the propagation phase, the expected incremental degree of each node in $V^{r}$ under each resampled particle is computed to predict the network evolution. In the $r$-th trial, the number of particles is $M$ and the number of nodes under each particle is scaled as $O(|V^{r}|)$, thus the network evolution learning algorithm {\bf Evo-NE} in $r$-th trial costs $O(M|V^{r}|)$ time.

\section{Learning  Evolving Influences}\label{learning}

Section \ref{Modeling Evolving} has illustrated the first step in $\mathbb{EIM}$ for learning network evolution. Now, we move to the second step of $\mathbb{ EIM}$ framework illustrated in Figure \ref{EvoIMMoverview}. That is, we need to learn the unknown influences $I(S, G_{t+T})$ among users to facilitate the accurate influenced size estimation over the predicted network $G_{t+T}$. Our  methodology of influence learning is presented as below.

{\bfseries Evolving influences modeling.} In the $r$-th trial, the objective is to maximize the influenced size over the target network $G^{r+1}$.  For any node pair $u_{i}$ and  $u_{j}$ in $G^{r+1}$, let $e$ denote the edge between $u_{i}$ and  $u_{j}$, and let $w_{e,r}$ denote weight of edge from  $u_{i}$ to $u_j$ during time $[T^{r}, T^{r+1})$. Built upon the widely used IC model depicted by Definition \ref{IC}, $u_{i}$ can successfully activates $u_{j}$ with a probability equal to $w_{e,r}$ during time $[T^{r}, T^{r+1})$. 
However, the traditional IC model cannot be directly applied to determine the weights in {\bf EIM} problem since: (1) The weights of newly established edges remain unknown in advance; (2) The weights may exhibit random dynamics with network evolution. The reason behind is that real-world factors such as users' interests of propagated contextual information and the closeness of user relations may be dynamic in evolution \cite{retweet}. For example, new edges are established when users make new friends, and the weights of edges may strengthen over time until they become stable close friends. In contrast, a pair of partners may drift apart after their cooperation has ended. Thus the weights of edges in evolution may randomly become larger or lower over time with decaying fluctuations. To jointly consider such features and the periodical learning-deciding framework of $\mathbb{EIM}$, we discretize the variations of the edges' weights and characterize each weight as a Gaussian random walk presented below, where its fluctuation from $T^{r}$ to $T^{r+1}$ can be represented by a Gaussian noise added to $w_{e,r}$.

{\bfseries Evolving weights of edges.} Let $w_{e,r}$ denote the value of $w_{e}$ in the $r$-th trial. For a new edge $e$ that establishes during $(T^{r-1}, \,T^{r}]$,  under the Gaussian random walk mode, we let the initial value of the weight $w_{e}$ follow a Gaussian distribution with $w_{e, r-1}\sim \mathcal{N}( \overline{w}_{e, r-1}',\, \mathbf{\Sigma}_{e, r-1})=\mathcal{N}( \overline{w}_{0},\, \mathbf{\Sigma}_{0})$ and $r_{e, 0}=r-1$. Here, $\overline{w}_{e, r-1}'$ and $\mathbf{\Sigma}_{e, r-1}$ respectively denote  the mean and variance of $w_{e,r-1}$'s distribution. 
Then the variation of $w_{e}$ is defined with a Markov process as below
\vspace{-2mm}
\begin{equation}\label{weight update}
 w_{e, r}= w_{e, r-1}+\mathbf{v}_{e, r}, \quad  \mathbf{v}_{e, r}\sim \mathcal{N}(0,  \Delta\mathbf{\Sigma}_{e,r}),
 \vspace{-1mm}
\end{equation}where $\mathbf{v}_{e, r}$ denotes the Gaussian random noise to characterize the variation of $w_{e}$ in the $r$-th trial and $\Delta\mathbf{\Sigma}_{e, r}=\frac{ \mathbf{\Sigma}_{0}}{(r-r_{e, 0})^{k}} (k>0)$. 

For the above evolving influences, recall that in Section 3.2, under the bandit-based framework of $\mathbb{EIM}$, we treat them as the arms and leverage the edge-level feedbacks to update their esimated values. In detail, let $w_{e}$ denote the weight of edge between $u_{i}$ and $u_{j}$. For a user $u_{i}$ being influenced in the $r$-th trial, he will try to influence his neighbor $u_{j}$ successfully with probability $w_{e,r}$, thus edge $e$ is triggered. Since a successful influence can bring a more influenced user, we model the reward obtained from edge $e$ as a binary reward $z_{e,r}$ with success denoted by $1$ and failure denoted by $0$, which is leveraged as the feedback to refine the distribution of $w_{e,r}$. Since the weight $w_{e}$ changes over time with a Gaussian random walk, in the $r$-th trial, it follows a Gaussian distribution after accumulating pervious random walks, which is denoted by $w_{e,r}\sim \mathcal{N}(\overline{w}_{e, r}',  \mathbf{\Sigma}_{e,r})$.
Thus, to estimate the real value of the weights provided with the Gaussian statistical properties, we adopt the Kalman Filter as the refining method for the distributions of evolving influences, as described below.

{\bfseries Kalman Filter based refining method.}
Let a binary variable $z_{e,r}$ denote the reward obtained from the triggered edge $e$ in the $r$-th trial, 
referring Kalman Filter theory \cite{kalman}, the mean $\overline{w}_{e, r}'$ and variance $\mathbf{\Sigma}_{e, r}$ of the weight of edge $e$ in the $r$-th trial is refined with
\vspace{-1mm}
\begin{align}
\label{mean}&\overline{w}_{e, r}'=\overline{w}_{e, r-1}'+\mathbf{G}_{e,r}\cdot(z_{e,r}-\overline{w}_{e, r-1}');\\
\label{square}&\mathbf{\Sigma}_{e, r}=\mathbf{\Sigma}_{e, r-1}+\Delta\mathbf{\Sigma}_{e, r}-\mathbf{G}_{e,r}\mathbf{\Sigma}_{e,r-1}.
\vspace{-1mm}
\end{align}Here, $\mathbf{G}_{e,r}(z_{e,r}-\overline{w}_{e, r-1}')$ and $\mathbf{G}_{e,r}\mathbf{\Sigma}_{e,r-1}$ are the correction from Kalman filte. And $\mathbf{G}_{e,r}$ is the Kalman Gain in refinement to quantify the correction from the new observation $z_{e, r}$, which is determined as follows.
 \vspace{-2mm}
\SetAlCapSkip{0.2em}
\begin{algorithm}[htbp]
	\begin{small}
		\begin{spacing}{0.8}
			\SetAlgoLined
			
		
			// \textit{Edge weight refining in the $r$-th trial}

\KwIn{Observed edges from time $T^{r-1}$ to $T^{r}$\;}
\KwOut{Refined distribution of each edge' s weight $w_{e,r}'$\;}
{\bfseries Process:} \\
Set  $w_{e, r}\sim \mathcal{N}( \overline{w}_{0},\, \mathbf{\Sigma}_{0})$ for each first observed edge;\\
\For{each observed edge $e$ in $E^{r}$}{
Compute $\mathbf{Q}_{e,r}=\mathbf{\Sigma}_{e, r-1}+1$\;
Compute $\mathbf{G}_{e,r}=\mathbf{\Sigma}_{e, r-1}\cdot \mathbf{Q}_{e,r}^{-1}$\;
Update $\overline{w}_{e, r}'=\overline{w}_{e, r-1}'+\mathbf{G}_{e,r}\cdot(z_{e,r}-\overline{w}_{e, r-1}')$\;
Update $\mathbf{\Sigma}_{e, r}=\mathbf{\Sigma}_{e, r-1}+\Delta\mathbf{\Sigma}_{e, r}-\mathbf{G}_{e,r}\mathbf{\Sigma}_{e,r-1}$\;
}
\For{each edge unobserved edge $e$ in $E^{r}$}{
Update $\overline{w}_{e, r}'=\overline{w}_{e, r-1}'$\;
Update $\mathbf{\Sigma}_{e, r}=\mathbf{\Sigma}_{e, r-1}+\Delta\mathbf{\Sigma}_{e, r}$\;
}
\end{spacing}
\Return  $\overline{w}_{e, r}'$, $\mathbf{\Sigma}_{e, r}$ for each edge.
\caption{{ Evolving influence learning ({\bfseries Evo-IL}).}}\label{iib}
	
		\end{small}

\end{algorithm} 

\begin{lemma}\label{Kalman Gain}
 The Kalman Gain in the refinement of $w_{e}$ in the $r$-th trial is determined by
 \begin{displaymath}
\mathbf{G}_{e,r}=\mathbf{\Sigma}_{e, r-1}\cdot \mathbf{Q}_{e,r}^{-1},
 \end{displaymath}
where $\mathbf{Q}_{e,r}=\mathbf{\Sigma}_{e, r}+1$ denotes variance of the activating result via $e$.
\vspace{-1mm}
\end{lemma}
 
The proof for Lemma \ref{Kalman Gain} is deferred to Appendix \ref{proof for Kalman Gain}.

On the other hand, when user $u_{i}$ is not influenced, edge $e$ is not triggered. Then the distributions for the non-triggered edges in the $r$-th trial evolve as: 
\begin{equation}\label{influence evolving}
\overline{w}_{e, r}'=\overline{w}_{e, r-1}', \,\mathbf{\Sigma}_{e, r}=\mathbf{\Sigma}_{e, r-1}+\Delta\mathbf{\Sigma}_{e, r}.
\end{equation}
The Kalman filter is used to refine the weights distributions of triggered edges, while the distributions for the non-triggered edges in the $r$-th trial evolve as Eqn. (\ref{influence evolving}). Combining the two cases of both triggered and non-triggered edges, Algorithm \ref{iib} shows the the  pseudo code for the evolving influence learning in $r$-th trial. It takes the activating results via each triggered edge as the reward to refine weights distributions of them, and outputs the updated weights distributions of edges in $E^{r}$. 

{\bf Complexity.} In each trial, the evolving influence learning algorithm {\bf Evo-IL} needs traverse both the triggered and non-triggered edges to update their distributions. Thus it costs $O(|E^{r}|)$ in the $r$-th trial where $E^{r}$ is the set of edges in $G^{r}$. 

According to the updated distributions, we adopt the Upper Confidence Bound (UCB) method to derive the estimating value of $w_{e,r}$, which is the reward distribution of arm $e$ in the $r$-th trial. In the traditional UCB framework \cite{semi-bandit}, given the mean value of the reward $A$ of an arm, and its variance $\Sigma$, its estimating value is determined by $A^{t}=A+c\sqrt{\Sigma}$. Accordingly, the estimating value of $w_{e}$ in the $r$-th trial is expressed as Definition \ref{we}. 

\begin{definition}\label{we}
The estimating value of $w_{e}$ in $r$-$th$ trial is:
\vspace{-1mm}
\begin{equation}\label{we1}
w_{e,r}'= \overline{w}_{e, r}'+c\sqrt{\mathbf{\Sigma}_{e, r}},
\vspace{-1mm}
\end{equation}where $c$ is a constant algorithm parameter in Linear generalization of UCB (LinUCB) \cite{semi-bandit}.
\end{definition}

{\bfseries Remark.}  Since the weight of each edge follows the Gaussian random walk, in the case when an edge is not observed in any trial, its variance increases with the number of trials as shown in Eqn. (\ref{weight update}).  At the same time, once an edge is observed in a trial, the distribution of edge weight distribution is refined by {\bf Evo-IL} from the reward of triggering the edge and the variance of the distribution is refined with Eqn. (\ref{square}). Since $\mathbf{\Sigma}_{e, r-1}+\Delta\mathbf{\Sigma}_{e, r}-\mathbf{G}_{e,r}\mathbf{\Sigma}_{e,r-1}<\mathbf{\Sigma}_{e, r-1}+\Delta\mathbf{\Sigma}_{e, r}$, the variance after the refinement is smaller than that in the unobserved case. Thus, according to the designing principle of UCB algorithm, the item $c\sqrt{\mathbf{\Sigma}_{e, r}}$ in Eqn (\ref{we1}) decreases with the number of observations.

\section{Evolving Seeds Selection}\label{evolving influence maximization}
Together with our solutions of learning network evolution in Section \ref{Modeling Evolving} and learning evolving influences in Section \ref{learning}, we are now able to embark on the evolving seed selection, which corresponds to the third step in each trial depicted in Figure \ref{EvoIMMoverview}.

\vspace{-1mm}
\subsection{Seeds Selection: Problem Reformulation}
As stated earlier, the seed selection in classical IM and {\bf EIM} problems are both NP-hard. And the proposed framework $\mathbb{EIM}$ aims at coping with the limitations of classical IM arisen from network evolution during influence diffusion. Note that our design in evolving seeds selection is not to jettison the pervious efforts in classical IM, but instead leverage the benefits of them wherever possible.
Thus, following the assumption in classical IM, the objective of seeds selection in $r$-th trial is to maximize the influence $\mathbb{E}(I(S^{r}, G^{r}))$ diffused to existing users in $G^{r}$. Meanwhile, in {\bf EIM}, the objective becomes maximizing $\mathbb{E}(I(S^{r}, G^{r+1}))$. However, the network structure of $G^{r+1}$ remains unknown in advance, and the known users are those having been observed until time $T^{r}$. To tackle this dilemma, we first generate an intermediate graph $\mathbb{G}^{r}$ so that the network evolution from $T^{r}$ to $T^{r+1}$ can be captured with prediction. Here, $\mathbb{G}^{r}$ is called as Intermediate Evolving Graph whose set of users are  those having known until time $T^{r}$, and each of them is attached with a weight that quantifies his influences to future users. The influences among users in $\mathbb{G}^{r}$ are estimated as Definition \ref{we} shown in Section \ref{learning}. Thus the {\bf EIM} problem becomes to maximize the sum of the weights of influenced existing users. This enables us to leverage the benefits of the well-studied classical IM techniques.  The Intermediate Evolving Graph is elaborated as follows.

{\bf Intermediate Evolving Graph $\mathbb{G}^{r}=(\mathbb{V}^{r}, \mathbb{E}^{r})$.}  The objective of generating $\mathbb{G}^{r}$ is to capture the network evolution with the weights attached to the existing known users. Before we construct the intermediate graph, we give the description of the influence diffusion process along with network evolution. Note that, under the IC model, each user only has a single chance to influence his neighbors after being influenced. For a user $u_{i}$ in $\mathbb{G}^{r}$, as described in Section \ref{Learning nodes growing speed}, he will expectedly have $\mathbb{E}(\Delta d_{i}^{r+1})$ potential neighbors until time $T^{r+1}$, with all the neighbors possibly influenced in the $r$-th trial if $u_{i}$ is influenced in $r$-th trial.
Thus, there are two possible cases of the $\mathbb{E}(\Delta d_{i}^{r+1})$ new neighbors: (1) if $u_{j}$ connects with $u_{i}$ before $u_{i}$ is influenced, then $u_{j}$ will be influenced with a probability of $w_{e,r}$ in the $r$-th trial; (2) if $u_{j}$ connects with $u_{i}$ when $u_{i}$ is influenced during the survival time of an IM campaign, $u_{j}$  will also have the chance to be influenced since, in reality, a newly acquainted friend on Twitter may sometimes review the tweets recently made. Since the edges between $u_{i}$ and such $\mathbb{E}(\Delta d_{i}^{r+1})$ potential neighbors have never been triggered, the weights of them are all estimated as the initial value, i.e., $\left(\overline{w}_{0}+c\sqrt{\mathbf{\Sigma}_{0}}\right)$. By above analysis, if $u_{i}$ becomes influenced in the $r$-th trial, he can further bring an expected number of  $\left(\overline{w}_{0}+c\sqrt{\mathbf{\Sigma}_{0}}\right) \cdot\mathbb{E}(\Delta d_{i}^{r+1})$ influenced users into account. Therefore, we set the weight of $u_{i}$ as $C_{i,r}=\mathbb{E}(\Delta d_{i}^{r+1}) \cdot \left(\overline{w}_{0}+c\sqrt{\mathbf{\Sigma}_{0}}\right)+1$, where the $1$ represents himself. Upon attaching the weight of each user over the Intermediate Evolving Graph, we reformulate the {\bf EIM} problem in Definition \ref{problem statement} as:

{\bfseries Problem Statement.} Let $I(S,  v_{i}, \mathbb{G}^{r})$ denote the probability that seed set $S$ can influence $v_{i}$ under IC model over the known structure of $\mathbb{G}^{r}$, where $v_{i}$ represents user $u_{i}$. Since $v_{i}$ can influence both himself and his potential neighbors with an expected number being $C_{i,r}$, the influence of $S$ via $v_{i}$ is equal to $I(S,  v_{i}, \mathbb{G}^{r})C_{i,r}$. Thus the expected influence of $S$ on the whole network $\mathbb{G}^{r}$ can be computed as $I(S, \mathbb{G}^{r})=\sum_{v_{i}\in \mathbb{V}^{r}}I(S,  v_{i}, \mathbb{G}^{r})C_{i, r}$. And the objective of evolving seed selection over $\mathbb{G}^{r}$ becomes
\vspace{-1mm}
\begin{equation}\label{problem}
S^{r}_{opt}=\mathop{\arg\max}_{S\subseteq \mathbb{V}^{r}}I(S, \mathbb{G}^{r}), \, |S|=K.
\vspace{-1mm}
\end{equation}
Before introducing the solution to Eqn. (\ref{problem}), we first demonstrate its key properties as stated in Lemma \ref{submodular}, which enables us to resolve it with performance guarantee. The corresponding proof of Lemma \ref{submodular} is available in Appendix \ref{hardnessproof2}. 
\vspace{-2mm}
\begin{lemma}\label{submodular}
The influence function in Eqn. (\ref{problem}) is monotonous and submodular.
\end{lemma}
%

\subsection{Seed Selection: Algorithm Design}

Over the intermediate evolving graph $\mathbb{G}^{r}$, we leverage the Influence Maximization via Martingale (IMM) framework to solve the {\bf EIM} problem transformed in Eqn. (\ref{problem}), which focuses on estimating the influence diffusion size of a given seed set $S$ over a general graph $G$ (i.e., $\mathbb{E}(I(S, G))$ via the  Reverse-Reachable Sets (RR-sets). The RR-sets \cite{IMM} is currently the most efficient way to resolve the classical IM problems and has been adopted by many IM techniques. Under the IMM framework of interests, RR-sets are utilized to largely improve the efficiency in estimating the influence diffusion size while still achieving the near optimal solution to classical IM problem.

 {\bf RR-sets.} Let $v$ be a given node in a general graph $G$, the RR-set for $v$ is the set of nodes that can reach it through active paths over $G$, which is generated as follows.  A deterministic copy $g$ of $G$ is firstly sampled, in which each edge $e$ is active with probability $w_{e}$ and inactive with probability $1-w_{e}$. Then the RR-set $R_{v}$ for node $v$ is generated by including into $R_{v}$ all the nodes that can reach $v$ via a backward Breadth-First Search (BFS) from $v$, and $v$ is treated as the root node of $R_{v}$. The key property of RR-set is that the probability that a seed set $S$ can influence a node $v$ over $G$ equals to the probability that $S$ overlaps $R_{v}$ \cite{IMM}. Thus given a randomly chosen node $v$, the expected influence of $S$ on $v$ is $\mathbb{E}[\mathbb{I}(S\cap R_{v} \neq \emptyset)]$, where $\mathbb{I}(\cdot)$ is the indicator function. Next, we will briefly review the general IMM framework \cite{IMM} before illustrating our solution to the seed set selection in {\bf EIM} problem.


\SetAlCapSkip{0.2em}
\begin{algorithm}[h]
	\begin{small}
			\SetAlgoLined
			
			// \textit{Influence maximization in evolving social networks}

\KwIn{Generated graph $\mathbb{G}^{r}$, number of selected seeds $K$\;}
\KwOut{A seeds set $S^{r}$\;}  

$l'=l\cdot(1+\log 2/ \log n)$\;
$\mathcal{R}$\,=\,Sampling $(\mathbb{G}^{r}, K, \varepsilon, l')$\;
$S^{r}$\,=\, NodeSelection $(\mathcal{R}, K)$\;
\Return $S^{r}$.
		\end{small}
		\caption{Evolving influence maximizaton ({\bfseries Evo-IMM}). }\label{EvoIMM}		
	\end{algorithm}
\begin{algorithm}[htbp]
	\begin{small}
		\begin{spacing}{0.8}
			\SetAlgoLined
		
			\KwIn{ Nodes in  intermediate evolving graph: $\mathbb{V}^{r}$ , ERR-sets: $\mathcal{R}$\;}
\KwOut{Sampled node $v$\;}
Initialize $n'=0$ and $\lambda_{0}=0$\;
\For{ each $v_{e}\in \mathbb{V}^{r}\backslash \mathcal{R}_{root}$}{
$\lambda_{e}=C_{e,r}$,\, $n'=n'+C_{e,r}$\;}

Divide interval $\big[0,\, n'\big]$ into $\big[0, \,\lambda_{1}\big], \big[ \lambda_{1},\, \lambda_{1}+\lambda_{2}\big],...,\big[
\sum_{i=1}^{|\mathbb{V}^{r}\backslash \mathcal{R}_{root}|-1}\lambda_{i}, \,n'\big] $\;
Randomly sample a constant $\alpha$ from interval $[0, \, 1]$\;
$n_{\alpha}=n'\cdot \alpha$\;
\If{$\sum_{j=0}^{e-1}\lambda_{j}\leq n_{\alpha}\leq\sum_{j=0}^{e}\lambda_{j}$}{$v=v_{e}$\;}
\end{spacing}
\Return node $v$.
					\end{small}
		\caption{Priority-based sampling ($\mathbb{V}^{r}, \mathcal{R}$)}\label{Priority}
		
	\end{algorithm}
	
\SetAlCapSkip{0.2em}
\begin{algorithm}[h]
	\begin{small}
		\begin{spacing}{0.8}
			\SetAlgoLined
\KwIn{Sampled ERR-sets $\mathcal{R}$,  number of selected seeds $K$\;}
\KwOut{A seed set $S^{r}$\;}

Initialize a seed set $S^{r}=\emptyset$\;
\For{k=1:K}{
Identify the node $v_{e}$ that maximizes $F_{\mathcal{R}}(S^{r}\cup v_{e})-F_{\mathcal{R}}(S^r)$\;
$S^r=S^r\cup \{v_{e}\}$\;
}
\end{spacing}
\Return $S^{r}$.

		\end{small}
		\caption{NodeSelection ($\mathcal{R}, K$)}\label{NodeSelection}	
	\end{algorithm}


{\bfseries General IMM.} The general IMM framework consists of two phases, i.e., sampling and node selection. The former phase iteratively generates a sufficiently large number of random RR-sets  to ensure the accuracy of influence estimation. And the latter one greedily selects a seed set of size $K$ to maximize the number of covered RR-sets. 
Let $\mathcal{R}=\{R_{1}, R_{2},..., R_{\theta}\}$ denote the generated RR-sets with the corresponding root nodes set being $\mathcal{R}_{root}$, and  $x_{1}, x_{2},..., x_{\theta}$ be the binary random variables denoting whether or not the corresponding RR-set is covered by the selected seeds set $S$. Then the influence of $S$ can be estimated by 
\vspace{-1mm}
{\small \begin{equation}\label{IMMequation}
\mathbb{E}[I(S)]=\frac{n}{\theta}\cdot \mathbb{E} \left(\sum_{i=1}^{\theta}x_{i}\right), 
\vspace{-1mm}
\end{equation}}where $n$ is the number of nodes in the networks. By Chernoff Bound, $\mathbb{E}[I(S)]$ can accurately estimate the influence of $S$ if $\theta$ is sufficiently large. 

 Borrowing the idea of IMM, our evolving seed selection algorithm {\bf Evo-IMM} is presented as follows. 

{\bf Evolving IMM algorithm ({\bfseries Evo-IMM}).} The key idea of {\bf Evo-IMM} is to apply the general IMM framework over the Generated Evolving Graph $\mathbb{G}^{r}$ for selecting the seed users in the $r$-th trail. Since each node in $\mathbb{G}^{r}$ is attached with a weight to quantify its influence to potential users, the RR-sets sampled from  $\mathbb{G}^{r}$ is correspondingly attached with a weight that equals to the weight of its root node. We call such set as ERR-set.
Let $n'=\sum_{v_{e}\in\mathbb{V}^{r}}C_{e,r}$ be the weighted sum of nodes in $\mathbb{V}^{r}$, and let $\theta' =\sum_{v_{e}\in\mathcal{R}_{root}}C_{e,r}$ denote the weighted sum of root node in ERR-sets. By Eqn. (\ref{IMMequation}) and the linearity of the expectation, the influence of a seed set $S$ over the generated graph $\mathbb{G}^{r}$ can be estimated as 
\vspace{-1mm}
{\small
\begin{equation}\label{EIMMequation}
\mathbb{E}[I(S, \mathbb{G}^{r})]=\frac{n'}{\theta'}\cdot \mathbb{E} \left(\sum_{i=1}^{\theta}x_{i}C_{i,r}\right). 
\vspace{-1mm}
\end{equation}}

Algorithm \ref{EvoIMM} shows the basic steps of {\bfseries Evo-IMM} in the $r$-th trial, in a same manner to the  general IMM. However, considering the influence of existing users in $\mathbb{G}^{r}$ on potential users, there are two major differences between {\bfseries Evo-IMM} and the general IMM. \underline{Firstly}, in sampling phase, a priority-based sampling method (Algorithm \ref{Priority}) is proposed to preferentially samples the ERR-sets whose root nodes have higher weights. Given the selected RR sets $\mathcal{R}$ and their root nodes set $\mathcal{R}_{root}$, the nodes in $\mathbb{V}^{r}\backslash \mathcal{R}_{root}$ with higher weights have higher probabilities to be sampled as next root nodes.
\underline{Secondly}, the Nodeselection phase (Algorithm \ref{NodeSelection}) in {\bfseries Evo-IMM} focuses on selecting the seed set with the maximum sum of weights of covered ERR-sets. Let $F_{\mathcal{R}}(\cdot)$ denote the weighted sum of covered ERR-sets. Nodeselection iteratively selects $K$ nodes with the maximum marginal gain to maximize $F_{\mathcal{R}}(S^{r})$. 

Based on the analysis in the general IMM framework, we refine the detailed settings of {\bf Evo-IMM}  in the way that it can meet both the high effectiveness and efficiency in seeds selection.
Algorithm \ref{Sampling} presents the Sampling phase in {\bfseries Evo-IMM}  which focuses on sampling enough number of ERR-sets to guarantee the accuracy for estimating $\mathbb{E}[I(S, \mathbb{G}^{r})]$. The parameters, i.e., $\varepsilon'$ and $\theta_{i}$ in Algorithm \ref{Sampling} are set based on the following lemma derived from IMM \cite{IMM} with the aim of ensuring the accuracy of influence estimation.
\SetAlCapSkip{0.2em}
\begin{algorithm}[h]
	\begin{small}
			\SetAlgoLined
			
			// \textit{ERR-sets sampling in line 2 of {\bfseries Evo-IMM} }

\KwIn{Intermediate evolving graph $\mathbb{G}^{r}$ , number of selected seeds $K$, error quantization parameters  $\varepsilon, l'$\;}
\KwOut{ERR-sets $\mathcal{R}$\;}  

Initialize a set $\mathcal{R}=\emptyset$ and a parameter $LB=1$, $LR=0$\;
$\varepsilon'=\sqrt{2}\cdot \varepsilon$\;
\For{$i=1: (\log_{2}n -1)$}{
$x=n'/ 2^{i}$, $\theta_{i}=\frac{(2+\frac{2}{3}\varepsilon')(\log \binom{n}{K}+l\cdot\log n+\log\log_{2}n)}{\varepsilon'^{2}\cdot x}$, $\theta'=0$\;
\While{$LR\leq \theta_{i}$}{
$v_{e}$=Priority-based sampling ($\mathbb{V}^{r}, \mathcal{R}$)\;
Generate the ERR-set for $v_{a}$ and insert it into $\mathcal{R}$\;
$LR=LR+1$, $\theta'=\theta'+C_{e,r}$\;
}
$S_{i}$\,=\,NodeSelection($\mathcal{R}, K$)\;
\If{$\frac{n'}{\theta'}F_{\mathcal{R}}(S_{i})\geq (1+\varepsilon')\cdot x$}{$LB=\frac{n'}{\theta'}F_{\mathcal{R}}(S_{i})/(1+\varepsilon')$\;
{\bfseries Break} \;}
}
$\theta= \frac{2n'((1-1/e)\cdot\alpha+\beta)^{2}}{LB\cdot\varepsilon^{2}}$, $\theta'=0$\;
\While{$\theta'\leq \theta$}{
$v_{e}$=Priority-based sampling ($\mathbb{V}^{r}, \mathcal{R}$)\;
Generate the ERR-set for $v_{a}$ and insert it into $\mathcal{R}$\;
$\theta'=\theta'+C_{e,r}$\;}

\Return $\mathcal{R}$.
		\end{small}
		 \caption{{Sampling $(\mathbb{G}^{r}, K, \varepsilon, l')$ in {\bf Evo-IMM}. }}\label{Sampling}
		
	\end{algorithm}

\begin{lemma}\label{IMM12}
In Algorithm \ref{Sampling}, define $\varepsilon_{1}=\varepsilon\frac{\alpha}{(1-1/e)\cdot \alpha+\beta}$ where
 {\small \begin{align} \alpha&=\sqrt{l'\log n+ \log 2} ,
 \\  \textnormal{and } \beta&=\sqrt{(1-1/e)\cdot\left(\log  \binom{n}{K}+l'\log n+\log 2\right) }.
\end{align}}Then with at least $(1-\frac{1}{n^{l'}})$ probability, the number of generated ERR-sets in sampling phase satisfies 
\begin{equation}\label{R}
|\mathcal{R}|\geq \frac{(2-2/e)\cdot n'\cdot \log(\binom{n}{K}\cdot 2n^{l'})}{(\varepsilon-(1-1/e)\cdot\varepsilon_{1})^{2}\cdot OPT} \quad(Theorem \, 2 \, in\,\cite{IMM}).
\end{equation}Suppose Inequality (\ref{R}) holds. By the properties of greedy algorithms, with at least  $(1-\frac{1}{2n^{l'}})$ probability, the returned set $S^{r}$ satisfies  
\begin{equation}\label{R1}
\frac{n'}{\theta'} F_{\mathcal{R}}(S^{r})\geq (1-1/e)(1-\varepsilon_{1})\cdot OPT \quad (Lemma\, 3 \, in\,\cite{IMM}),
\end{equation}where $OPT$ denotes the weighted sum of expected influenced nodes by the optimal seeds set with size $K$. 
\end{lemma}

Lemma \ref{IMM12} indicates that {\bf Evo-IMM} samples a sufficient number of ERR-sets and returns a seed set which covers a large number of ERR-sets. The weighted sum of covered ERR-sets (i.e., $\frac{n'}{\theta'} F_{\mathcal{R}}(S^{r})$ in Eqn. (\ref{R1})) serves as the indicator of expected influence  (i.e., $\mathbb{E}[I(S^{r}, \mathbb{G}^{r})]$), which guarantees the effectiveness of seed selection in $\mathbb{EIM}$.
Next, we will present the theoretical performance guarantee of {\bf Evo-IMM} in seed selection from the perspective of both effectiveness and efficiency.

\subsection{Performance Analysis of Seed Selection.}Lemma \ref{IMM12} lays the foundation for analyzing the effectivenss of {\bfseries Evo-IMM}, and the detailed analysis is shown in Lemma \ref{approximation ratio}, with the proof shown in Appendix \ref{hardnessproof2}. 

\begin{lemma}\label{approximation ratio}
If the Inequalities (\ref{R}) and  (\ref{R1}) hold, with at least $(1-1/2n^{l'})$ probability, we have $\mathbb{E}[I(S^{r}, \mathbb{G}^{r})]\geq (1-1/e-\varepsilon)\cdot OPT$. 
\vspace{-1MM}
\end{lemma}

Lemme 6.2 shows the accuracy of influence estimating and Lemme 6.3 shows the approximating ratio for the Nodeselection phase. Combing Lemmas 6.2 and 6.3, we can demonstrate the effectiveness of {\bf Evo-IMM}, as stated by Corollary \ref{probability}.

\begin{corollary}\label{probability}
Based on the union bound, {\bfseries Evo-IMM} returns a $(1-1/e-\varepsilon)$ approximate seeds set $S^{r}$ to the evolving IM problem with a probability of at least $1-1/2n^{l'}-1/2n^{l'}-1/n^{l'}=1-2/n^{l'}=1-1/n^{l}$.
\end{corollary}
%

Corollary 6.4 manifests that {\bfseries Evo-IMM} can find a nearly optimal solution for the evolving IM problem with a high probability. Lemma \ref{complexity} further provides the polynomial time complexity that Evo-IMM enjoys. 

\begin{lemma}\label{complexity}
The time complexity of {\bfseries Evo-IMM} is $O\big((K+l) \big((n+m)+\frac{n}{OPT}\big)\log n/\varepsilon^{2}\big)$, where $n=|\mathbb{V}^{r}|$ and $m=|\mathbb{E}^{r}|$. 
\end{lemma}

The proof for Lemma \ref{complexity} is shown in Appendix \ref{proof for efficiency}. 
 
 Based on Corollary 6.4 and Lemma \ref{complexity}, we draw the conclusion that {\bf Evo-IMM} can efficiently solve the seeds selection in {\bf EIM} and simultaneously enjoys comparable approximation ratio and time costs to the general IMM framework in static networks, as stated in Theorem \ref{ratio}. 

\begin{theorem}\label{ratio}
 {\bfseries Evo-IMM} returns a $(1-1/e-\varepsilon)$-approximate seed set to {\bf EIM} problem with a probability of at least $(1-1/n^{l})(l\geq 1)$, and it runs in $O\left((K+l)((|\mathbb{V}^{r}|+|\mathbb{E}^{r}|)+\frac{|\mathbb{V}^{r}|}{OPT})\log |\mathbb{V}^{r}|/\varepsilon^{2}\right)$,  where $OPT$ refers to the expected influenced size of the optimal seed set.
 \end{theorem}
While Theorem \ref{ratio} summarizes the performance guarantee of  {\bfseries Evo-IMM}, Corollary \ref{overall complexity} derives  the complexity of the $r$-th trial in $\mathbb{EIM}$.
\begin{corollary}\label{overall complexity}
Together with the three steps as illustrated in Sections 4, 5 and 6, the complexity of $r$-th trial in $\mathbb{EIM}$ is $O\big(M|\mathbb{V}^{r}|+|\mathbb{E}^{r}|+(K+l)((|\mathbb{V}^{r}|+|\mathbb{E}^{r}|)+\frac{|\mathbb{V}^{r}|}{OPT})\log |\mathbb{V}^{r}|/\varepsilon^{2}\big)$. Here, $M$ is the number of particles.
 \vspace{-2mm}
\end{corollary}

 Notably, in each trial, the seeds returned by {\bf Evo-IMM} are selected under the influences represented by the estimated values in Definition \ref{we} since the real value of $w_{e,r}$ remains unknown in advance. Here, how is the quality of the selected seeds? And what is the gap between its quality and that selected under fully known influences? In the sequel, we answer those questions via performance analysis of $\mathbb{EIM}$.

\section {Performance Analysis of $\mathbb{EIM}$}\label{Performance Analysis}
Recall again that in Section 3.2, $\mathbb{EIM}$ is a bandit-based framework where the arms represent evolving influences refined in a manner depicted by Definition \ref{we}. Thus to demonstrate its theoretical performance guarantee, we provide the analysis for its {\emph Regret}, which is quantified by the loss of influenced size incurred by the bandits. Let $I(S_{opt}^{r}, G^{r+1})$ denote the expected influenced size of the seeds selected under the ideal condition that influences among users are all known, and let $I(S^{r}, G^{r+1})$ denote that of the seeds selected under the estimating values. Intuitively, the regret of the bandits in the $r$-th trial is equal to $I(S_{opt}^{r}, G^{r+1})-I(S^{r}, G^{r+1})$.  However, recall Lemma \ref{hardness}, the {\bf EIM} problem is NP-hard and its objective function is submodular. As a result, even under the ideal condition, the selected seed set can only be a suboptimal one and achieve an expected influenced size of $\beta I(S_{opt}^{r}, G^{r+1})$. Here $\beta$ is the approximating ratio in seeds selection. 
Thus the regret incurred by the bandits is defined as the scaled cumulative regret \cite{IMbandit} as follows:
 \vspace{-2mm}
\begin{definition}\label{regret bound 1}
(Scaled regret.) Given the approximating ratio of IM algorithm in step (3) is $\beta$, the regret $B$ over $R$ trials is equal to
\vspace{-1mm}
\begin{equation}
\mathbb{E}(B)=\sum_{r=1}^{R}\left( I(S_{opt}^{r}, G^{r+1})- \frac{1}{\beta}I(S^{r}, G^{r+1})\right),
 \vspace{-1mm}
\end{equation}where $S_{opt}^{r}$ is the optimal seeds set in the $r$-th trial and $S^{r}$ is the seeds set returned by $
\mathbb{ EIM}$.
\end{definition}
Based on Theorem \ref{ratio}, {\bf Evo-IMM} can return a $(1-1/e-\varepsilon)$-approximate solution with a probability of more than $(1-1/n^{l})$, thus $\beta=(1-1/e-\varepsilon)\cdot(1-1/n^{l})$ in $\mathbb{EIM}$. Under Definition \ref{regret bound 1}, the regret bound of $\mathbb{EIM}$ is the upper bound of the gap between $\sum_{r=1}^{R}I(S_{opt}^{r}, G^{r+1})$ and $\sum_{r=1}^{R}\frac{1}{\beta}I(S^{r}, G^{r+1})$ as we will disclose in Theorem \ref{friendship bound} shortly. 
$\mathbb{EIM}$ focuses on selecting seed users with the predicted network evolution and refined influences through {\bf Evo-IMM}. The reward of selected seeds, which reflect their qualities, is the expected number of influenced users over the target network. Under the same seed selection algorithm and target network, the quality of the selected seeds is dominated by the accuracy of influences estimating. Thus the regret of $\mathbb{EIM}$ in the $r$-th trial is dominated by the estimating error of influences, whose distributions are refined by {\bf Evo-IL} (Algorithm \ref{iib}), over the target network $G^{r+1}$.  
\begin{theorem}\label{friendship bound}
The regret bound of  $\mathbb{EIM}$  can be scaled as
\vspace{-1mm}
\begin{equation}
\mathbb{E}(B)\leq O\left(\sqrt{|E_{R+1}|\ln(R+1) R}\right).
\vspace{-1mm}
\end{equation}Here, $|E_{R+1}|$ denotes the number of edges until time $T^{R+1}$.
\end{theorem}
\begin{proof}
We divide the whole proof into 4 steps.

{\bf 1. Overall regret bound over the $R$ trials.}

In the $r$-th trial, we denote vectors $\vec{w}_{r}$ and $\vec{w}_{r}'$ as the real and estimating weights of edges in $E^{r+1}$ respectively, where $|\vec{w}|=|\vec{w}_{r}'|=|E^{r+1}|$ since network evolves from $G_{r}$ to $G_{r+1}$ during the $r$-th trial. Correspondingly, $I(S, \vec{w}_{r})$ and $I(S, \vec{w}_{r}')$ represent the expected influence of seeds set $S$ under $\vec{w}_{r}$ and $\vec{w}_{r}'$ respectively. Since the distribution of $w_{e,r}$ is estimated as Definition \ref{we}, we define an event $\mathcal{F}_{r}$ corresponds to $w_{e,r}'$ as below
 \vspace{-1mm}
\begin{displaymath}
\mathcal{F}_r \triangleq \{|\overline{w}_{e,r}' - w_{e,r}|\leq c\sqrt{\mathbf{\Sigma}_{e,r}}, \,\forall e\in E^{r+1} \}.
\vspace{-1mm}
\end{displaymath}Under event $\mathcal{F}_r$, we have $0\leq w_{e,r}'-w_{e,r} \leq 2c\sqrt{\mathbf{\Sigma}_{e,r}}$, and
\begin{displaymath}
 I(S_{opt}^{r},\vec{w_{r}}) \leq I(S_{opt}^{r},\vec{w}_{r}') \leq \frac{1}{\beta} \mathbb{E}[f(S^{r},\vec{w}_{r}')].
 \end{displaymath}
Based on Definition \ref{regret bound 1},  the regret bound of $\mathbb{EIM}$ over $R$ trials can be formulated as 
\begin{align}
\notag &\mathbb{E}[B]=\sum_{r=1}^{R}I(S_{opt}^{r},\vec{w}_{r)} - \frac{1}{\beta} \mathbb{E}[I(S^r,\vec{w}_{r})]\\
\notag &\leq\underbrace{\sum_{r=1}^{R} \frac{1}{\beta} \mathbb{E}[I(S^r, \vec{w}_r')-I(S^r,\vec{w}_r)|\mathcal{F}_{r}]}_{L1} +\underbrace{\sum_{r=1}^{R}P(\overline{\mathcal{F}}_{r})|V^{r+1}|}_{L2}
\end{align}Here, $P(\overline{\mathcal{F}}_{r})|V^{r+1}|$ means the regret is no more than $|V^{r+1}|$ even under the worst case. Now, we continue to derive the respective upper bound of $L1$ and $L2$.

{\bf 2. Upper bound of $L1=\sum_{r=1}^{R}\frac{1}{\beta} \mathbb{E}[I(S^r, \vec{w}_r')-I(S^r,\vec{w}_r)|\mathcal{F}_{r}]$.}

To facilitate the computation of $I(S^r, \vec{w}_r')$ and $I(S^r,\vec{w}_r)$, we model the evolving network as evolving forest where we only take account of one path between a pair of users into the regret analysis.  The reason for adopting forest model is two foldes:  (1) under IC model, the influences diffused from seed nodes to other nodes in the network are expectedly along the path with maximum edge weights \cite{chen2010scalable}; (2) the influence diffusion is progressive in every IM campaign where once a node being influenced by seeds through the path with maximum edge weights, it will remain influenced permanently. Thus the forest which only takes account of one path between any pair of nodes is widely utilized in existing works for computing the expected influences of seeds with the most representative one being maximum influence arborescence (MIA) model \cite{chen2010scalable}\cite{ICDE16}.

   \begin{figure}[h]
  \vspace{-4mm}
 \centering
\centering
  \includegraphics[width=0.3\textwidth]{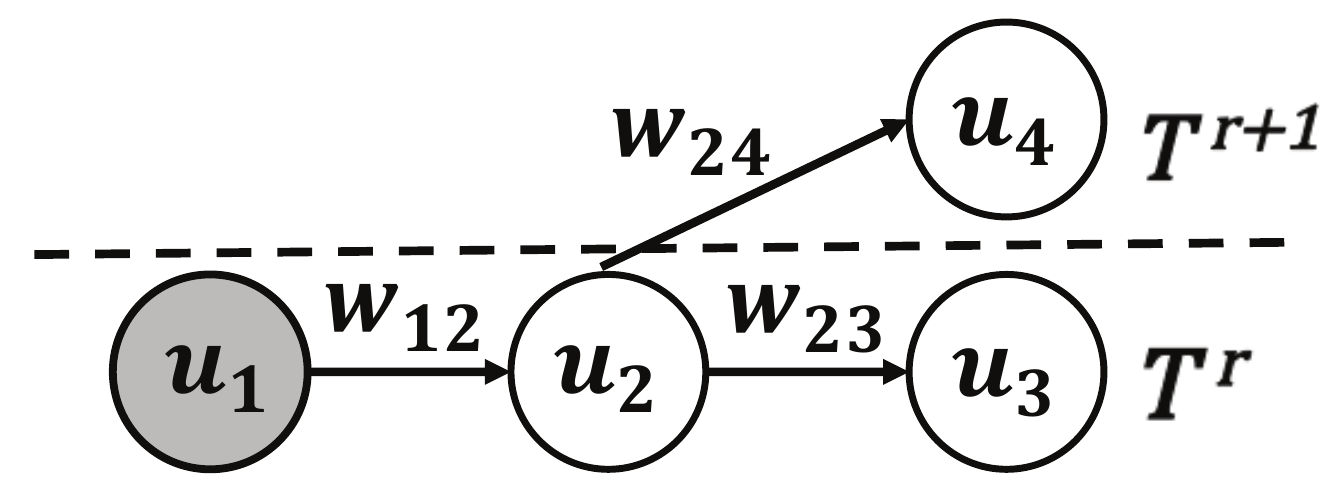}
  \vspace{-4mm}
  \caption{An illustrative example of evolving forest}\label{forest}
   \vspace{-3mm}
  \end{figure}
  
An illustrative example of evolving forest is shown in Figure \ref{forest}. At time $T^{r}$, there are three nodes $u_{1}$, $u_{2}$ and  $u_{3}$ in the network, which can be modeled as the traditional forest in static networks since there is only one path between any pair of nodes. For a new node $u_{4}$ joining during $T^{r}$ to $T^{r+1}$, there is still only one path between it and anyone existing users. Thus the network in evolution is still modeled as the forest, which is called as evolving forest. Under the forest structure, we can exactly compute the expected influence of a given seed node over the network. In case that node $u_{1}$ is the seed, its influence over the network at time $T^{r+1}$ is equal to $I(S, G^{r+1})=w_{12}+w_{12}w_{23} +w_{12}w_{24}$.

By modeling evolving network as the evolving forest, we only take account of one path from any seed in $S^r$ to the nodes belonging to $V^{r+1}\backslash S^r$.  Let $T_{r, v}$ denote set of paths from seeds in $S^r$ to node $v\in V^{r+1}\backslash S^r$. We can exactly compute $\mathbb{E}[I(S^r, \vec{w}_{r})]$ as a polynomial of the real weights of edges in $T_{r, v}$ and  compute $\mathbb{E}[I(S^r, \vec{w}_{r}')]$ as the polynomial of corresponding estimating values.  Thus, under the evolving forest, $\mathbb{E}\big[I(S^r, \vec{w}_{r}')- I(S^r,\vec{w}_r)\big]$ can be represented as a polynomial of errors in edge weights estimation (i.e., $(w_{e, r}' - w_{e, r}) (e\in E^{r+1}))$, which enables us to quantify the regret from the UCB based estimating values in Definition \ref{we}. Then $L1$   becomes 
\begin{align}
\notag L1&=\sum_{r=1}^{R}\frac{1}{\beta} \mathbb{E}[I(S^r, \vec{w}_r')-I(S^r,\vec{w}_r)|\mathcal{F}_{r}]\\
\label{L1} &\leq \frac{1}{\beta} \sum_{v\in V^{r+1}\setminus S^r}\sum_{e\in T_{r,v}} \mathbb{E}[\mathbb{I}(o_e^r)(w_{e,r}'-w_{e,r})]. 
\end{align}In Inequality (\ref{L1}), $\mathbb{I}(o_e^r)$ denotes the event that edge $e$ is triggered in the $r$-th trial, i.e, at lest one endpoint of edge $e$ is influenced in the $r$-th trial. We use $P(o_{e}^{r})$ to denote the probability of event $\mathbb{I}(o_e^r)$ with $P(o_{e}^{r})=\mathbb{I}(o_e^r)$.  The proof of Inequality (\ref{L1}) is presented in Appendix \ref{AC}. Since $0\leq w_{e,r}'-w_{e,r} \leq 2c\sqrt{\mathbf{\Sigma}_{e,r}}$, Inequality (\ref{L1}) further becomes 
\begin{equation}
\notag L1 \leq \sum_{r=1}^{R}\frac{2c}{\beta} \sum_{e\in E^{r+1}}\mathbb{E}\left[\mathbb{I}(o_e^r)N_{r,e}\sqrt{\mathbf{\Sigma}_{e,r}} \right]. 
\end{equation}Here, $N_{r,e}$ is the number of paths in set $T_{r,v}$ that contains edge $e$. Then according to   Cauchy-Schwarz Inequality, we have
{\small \begin{align}
\notag&\mathbb{E}\left[\sum_{r=1}^{R}\sum_{e\in E^{r+1}}\mathbb{I}(o_e^r)N_{r,e}\sqrt{\mathbf{\Sigma}_{e,r}}\right] \\
\label{dec} \leq& \sqrt{\mathbb{E}\left[\sum_{r=1}^{R}\sum_{e\in E^{r+1}}\mathbb{I}(o_e^r)N_{r,e}^2\right]\mathbb{E}\left[\sum_{r=1}^{R}\sum_{e\in E^{r+1}}\mathbb{I}(o_e^r)\mathbf{\Sigma}_{e,r} \right]}
\end{align}}Since the number of edges grows with network evolution, we define a network parameter $C$ in 
Inequality (\ref{dec}) as below to bound the effect of network size on the regret, which enables to explore the correlations between regret of $\mathbb{EIM}$ and the number of trials.  {\small \begin{equation}
\notag C \triangleq \max_{S^r:|S^r|=K, 1\leq r \leq R}\sqrt{\sum_{e \in E^{r+1}} N_{r,e}^{2} \cdot P(o_{e}^{r})}.
\end{equation}}
Then Inequality (\ref{dec}) becomes
{\small \begin{equation}
\label{a15}  \mathbb{E}\left[\sum_{r=1}^{R}\sum_{e\in E^{r+1}}\mathbb{I}(o_e^r)N_{r,e}\sqrt{\mathbf{\Sigma}_{e,r}}\right] 
 \leq C\sqrt{R}\sqrt{\sum_{r=1}^{R}\sum_{e\in E^{r+1}}\mathbb{I}(o_e^r)\mathbf{\Sigma}_{e,r}}.
\end{equation}}To give the upper bound of Eqn. (\ref{a15}), next we first provide the analysis of the term $\sum_{r=1}^{R}\mathbb{I}(o_e^r)\mathbf{\Sigma}_{e,r}$. 

{\bf Upper bound of $\sum_{r=1}^{R}\mathbb{I}(o_e^r)\mathbf{\Sigma}_{e,r}$. in Eqn. (\ref{a15}).}

(1) We first consider a special case where an edge is observed in all $R$ trials. Without loss of generality, we take edge $e$ as an example and use $\mathbf{\Sigma}_{r}$ to denote its variance in the $r$-th trial in the analysis of the upper bound of  $\sum_{r=1}^{R}\mathbb{I}(o_e^r)\mathbf{\Sigma}_{e,r}$. By Eqn. (\ref{square}), we have 
\begin{equation}\label{square1}
\mathbf{\Sigma}_{r+1}= \frac{\mathbf{\Sigma}_0^2}{(r+1)^k} + \frac{\mathbf{\Sigma}_{r}}{\mathbf{\Sigma}_{r}+1}.
\end{equation}Let $\mathbf{\Sigma}_0^2 \leq 1 \leq \frac{3}{r^{\frac{k}{2}}}|_{r=1}$ since $0\leq w_{r}\leq 1$. Referring to Lemma \ref{b1} (in Appendix H), if $\mathbf{\Sigma}_{r-1} \leq \frac{3}{(r-1)^{\frac{k}{2}}}$ and $k\leq 2$, we have 
\begin{displaymath}
\mathbf{\Sigma}_r\leq \Delta\mathbf{\Sigma}_r + \frac{\frac{3}{(r-1)^{\frac{k}{2}}}}{\frac{3}{(r-1)^{\frac{k}{2}}} + 1} \leq \frac{1}{r^k} + \frac{3}{(r-1)^{\frac{k}{2}}+3} \leq \frac{3}{r^{\frac{k}{2}}}. 
\end{displaymath} Hence, by induction, we can draw the conclusion that $\mathbf{\Sigma}_r\leq \frac{3}{r^{\frac{k}{2}}}$. And by Lemma \ref{b2} (in Appendix H), we have 
\begin{equation}\label{a21}
\sum_{r=1}^{R} \mathbf{\Sigma}_r \leq \sum_{r=1}^{R} \frac{3}{r^{\frac{k}{2}}}\leq  \frac{6}{2-k}R^{1-\frac{k}{2}}\, (0< k< 2). 
\end{equation}For $k \geq 2$, we have $\mathbf{\Sigma}_r\leq \frac{3}{r}$, and  $\sum_{r=1}^{R}\mathbf{\Sigma}_{r}$ becomes  
\begin{equation}\label{a16}
\sum_{r=1}^{R} \mathbf{\Sigma}_r  \leq \sum_{r=1}^{R} \frac{3}{r} \leq 3\ln T + 3. 
\end{equation}

(2) Now we consider the general case when the edge $e$ is not observed in at least one trial. Notably, if $e$ is not observed in the $r$-th trial, it is not counted into the regret computation according to Eqn. (\ref{a15}) since $\mathbb{I}(o_{e}^{r})=0$. 

We start with the case when edge $e$ is not observed in a single trial, i.e., the $\tau$-th trial. Let $\mathbf{\Sigma}_r'$ denote the variance of $w_{e,r}$ in this case, thus $\mathbf{\Sigma}_r=\mathbf{\Sigma}_r', \forall r\leq \tau$. By Eqn. (\ref{square}), and $\mathbf{\Sigma}_\tau \leq 1, \frac{\mathbf{\Sigma}_{\tau}}{\mathbf{\Sigma}_{\tau}+1} \geq \frac{\mathbf{\Sigma}_{\tau}}{2}$, we have
\begin{align}
\notag \mathbf{\Sigma}_{\tau+1} &= \Delta\mathbf{\Sigma}_{\tau+1} + \frac{\mathbf{\Sigma}_{\tau}}{\mathbf{\Sigma}_{\tau}+1},\\
\notag \mathbf{\Sigma}'_{\tau+1} &=  \Delta\mathbf{\Sigma}_{\tau+1}+ \mathbf{\Sigma}'_{\tau},\\
\label{a20} \mathbf{\Sigma}_\tau + \mathbf{\Sigma}_{\tau+1} - \mathbf{\Sigma}'_{\tau+1} &=  \frac{\mathbf{\Sigma}_{\tau}}{\mathbf{\Sigma}_{\tau}+1},\\
\label{a18} \mathbf{\Sigma}'_{\tau+1} - \mathbf{\Sigma}_{\tau+1} &= \mathbf{\Sigma}_{\tau} - \frac{\mathbf{\Sigma}_{\tau}}{\mathbf{\Sigma}_{\tau}+1}\leq \frac{\mathbf{\Sigma}_{\tau}}{\mathbf{\Sigma}_{\tau}+1}.
\end{align}According to Lemma \ref{b3} (in Appendix H), and assume $\mathbf{\Sigma}_\tau \leq \frac{3}{\tau^{\frac{k}{2}}} + \epsilon_\tau$, then
\begin{equation}
\begin{split}
\mathbf{\Sigma}_{\tau+1} & \leq  \frac{1}{(\tau+1)^k} + \frac{\frac{3}{\tau^{\frac{k}{2}}} + \epsilon_\tau}{\frac{3}{\tau^{\frac{k}{2}}} + \epsilon_\tau + 1 }\\
& = \frac{1}{(\tau+1)^k} + \frac{ \epsilon_\tau(r-1)^k+3}{( \epsilon_\tau+1)(r-1)^k + 3} \\
& \leq \frac{3}{(\tau+1)^{\frac{k}{2}}} + \frac{\epsilon_\tau}{4}
\end{split}	
\end{equation}By induction, we have $\mathbf{\Sigma}_{\tau+n} \leq \frac{3}{(\tau+n)^{\frac{k}{2}}} + \frac{\epsilon_\tau}{4^n}$. And according to Eqn. (\ref{a18}), $\mathbf{\Sigma}'_{\tau+1}$ satisfies 
\begin{displaymath}
\mathbf{\Sigma}'_{\tau+1} \leq \frac{3}{(\tau+1)^{\frac{k}{2}}} + \frac{\epsilon_{\tau}}{4} + \frac{\mathbf{\Sigma}_{\tau}}{\mathbf{\Sigma}_{\tau}+1}.
\end{displaymath}
Thus, based on the induction above, we have 
\begin{equation}\label{a20}
\mathbf{\Sigma}'_{\tau+n} \leq \frac{3}{(\tau+n)^{\frac{k}{2}}} + \frac{1}{4^n}\left(\epsilon_{\tau}+ \frac{4\mathbf{\Sigma}_{\tau}}{V_{\tau}+1}\right) 
\end{equation}
\begin{displaymath}
\sum_{n=2}^{R-\tau}\mathbf{\Sigma}'_{\tau+n} \leq \sum_{n=2}^{R-\tau} \left(\frac{3}{(\tau+n)^{\frac{k}{2}}} +\frac{\epsilon_{\tau}}{4^n}\right) + \frac{\mathbf{\Sigma}_{\tau}}{3(\mathbf{\Sigma}_{\tau}+1)}. 
\end{displaymath}Then, by Eqn. (\ref{a20}), for the variance of $w_{e}$ from the $(\tau+1)$-th trial to the $R$-th trial, we have 
\begin{equation}\label{a19}
\sum_{r=\tau+1}^{R}\mathbf{\Sigma}'_r \leq \sum_{r=\tau}^{R} \left(\frac{3}{r^{\frac{k}{2}}} +\frac{\epsilon_{\tau}}{4^{r-\tau}}\right). 
\end{equation}According to Inequality (\ref{a21}), $\Sigma_{r}\leq\frac{3}{r^{\frac{k}{2}}}$ holds, thus $\epsilon_{\tau}=0$. Then we have $\sum_{r=\tau+1}^{R}\mathbf{\Sigma}'_r \leq \sum_{r=\tau}^{R} \frac{3}{r^{\frac{k}{2}}}$.

Next, we consider the case when edge $e$ is not observed in $\tau_{1}, \tau_{2},..., \tau_{i}$-th trials. Let $\mathbf{\Sigma}'_{1,r}, \mathbf{\Sigma}'_{2,r},..., \mathbf{\Sigma}'_{i,r}$ denote the variance of $w_{e,r}$ when edge $e$ is not observed in $\{\tau_{1}\}, \{\tau_{1},  \tau_{2}\},..., \{\tau_{1}, \tau_{2},..., \tau_{i}\}$-th trials respectively. Specially, $\mathbf{\Sigma}'_{0,r}$ denotes the variance in case that edge $e$ is observed in all trials and $\mathbf{\Sigma}'_{0,r}=\mathbf{\Sigma}_{r}$. And based on the analysis in last subsection, we have 
\begin{equation}\label{a22}
\sum_{r=[R],r\neq \tau_1} \mathbf{\Sigma}'_{1,r} \leq \sum_{r=1}^{R} \mathbf{\Sigma}'_{0,r}.
\end{equation}
For $\mathbf{\Sigma}'_{1,\tau_{2}}$, we have the following inequality from Eqn. (\ref{a20}): 
\begin{displaymath}
 \mathbf{\Sigma}_{1,\tau_2}' \leq \frac{3}{\tau_2^{\frac{k}{2}}} +  \frac{ \mathbf{\Sigma}_{0,\tau}'}{4^{\tau_2-\tau_1-1}( \mathbf{\Sigma}_{0,\tau}'+1)}
\end{displaymath}
And similar to Eqn. (\ref{a19}), we also have
{\small
\begin{displaymath}
\sum_{r=\tau_2+1}^{R}\mathbf{\Sigma}'_{2,r} \leq \sum_{r=\tau_2}^{R} \left(\frac{3}{r^{\frac{k}{2}}} +\frac{\mathbf{\Sigma}'_{0,\tau}}{4^{r-\tau_1-1}(\mathbf{\Sigma}'_{0,\tau}+1)}\right)
\end{displaymath}
\begin{displaymath}
\begin{split}
\sum_{r=[\tau_1+1,\cdots,R],r\neq \tau_2}\mathbf{\Sigma}'_{2,t}& \leq \sum_{r=\tau_1+1}^{R} \left(\frac{3}{r^{\frac{k}{2}}} +\frac{\mathbf{\Sigma}'_{0,\tau}}{4^{r-\tau_1-1}(\mathbf{\Sigma}'_{0,\tau}+1)}\right)\\
& \leq \sum_{r=\tau_1}^{R} \frac{3}{r^{\frac{k}{2}}}
\end{split}
\end{displaymath}}Hence, 
\begin{equation}\label{a23}
\sum_{r\in [R]\setminus\{\tau_1,\tau_2\}} \mathbf{\Sigma}'_{2,r} \leq \sum_{r=1}^{R} \frac{3}{r^{\frac{k}{2}}}. 
\end{equation}
Therefore, by corresponding Eqn. (\ref{a22}) to $\mathbf{\Sigma}'_{1,r}$ and  Eqn. (\ref{a23}) to $\mathbf{\Sigma}'_{2,r}$, we can inductively draw the following conclusion:
\begin{equation}\label{a38}
\sum_{r\in [R]\setminus\{\tau_1,\cdots,\tau_i\}}  \mathbf{\Sigma}'_{i,r} \leq \sum_{r=1}^{R} \frac{3}{r^{\frac{k}{2}}}.
\end{equation}Taking $\sum_{r=1}^{R} \mathbb{I}(o_{e}^{r})\mathbf{\Sigma}_{e,r}\leq\sum_{r=1}^{R}\frac{3}{r^{\frac{k}{2}}}$ into Inequality (\ref{a15}), we can obtain the upper bound of $L1$ with $L1 \leq  O\left(\sqrt{R|E_{R+1}|\sum_{r=1}^{R} \frac{3}{r^{\frac{k}{2}}}} \right).$

{\bf 3. The upper bound of $L2=\sum_{r=1}^{R} P(\overline{\mathcal{F}}_{r})|V_{R+1}|$.}

We first review the definition of event $\mathcal{F}_{r}$, i.e., $\mathcal{F}_r \triangleq \{|w_{e,r}' - w_{e,r}|\leq c\sqrt{\mathbf{\Sigma}_{e,r}}, \,\forall e\in E^{r} \}$. In the $r$-th trial, the observing value of $w_{e,r}$ can be formulated as 
\begin{displaymath}
z_{e,r} = w_{e,r}' + \sigma_{e,r},
\end{displaymath}where $\sigma_{e,r}$ denotes the observing error with zero mean and $\sigma_{e,r}\in (-1,1),\forall e,r$. Then according to the Lemma \ref{subgaussian}, 
$\sigma_{e,r}$ follows the sub-gaussian distribution with a variance upper bounded by $1$. 
\begin{lemma}\label{subgaussian}
(\cite{sub}.) If $X$ is a random variable with $\mathbb{E}(X) = 0$ and $|X|\leq b$ a.s.
for some $b>0$, then $X$ is b-subgaussian.
\end{lemma}
Hence, we have 
\begin{equation}
\notag P(|w_{e,r}-w_{e,r}'| > c\sqrt{\mathbf{\Sigma}_{e,r}}) \leq e^{-\frac{c^2}{2}}.
\end{equation}Let $\bar{|E|}$ denote the mean number of edges in $R$ trials, then 
\begin{displaymath}
\sum_{r=1}^{R} P(\overline{\mathcal{F}}_{r})|V_{R+1}| \leq 2e^{-\frac{c^{2}}{2}}\cdot \bar{|E|}|V_{R+1}|R. 
\end{displaymath}In case that $c\geq 2\sqrt{\ln 2\bar{|E|}|V_{R+1}|R}$, we have  $\sum_{r=1}^{R} P(\overline{\mathcal{F}}_{r})|V_{R+1}|\leq 1$. 

{\bf 4. Conclusion.}

Together with upper bound of both $L1$ and $L2$, we can derive the regret bound of $\mathbb{EIM}$ over the $R$ trials. Let $c= 2\sqrt{\ln 2\bar{|E|}|V_{R+1}|R}$, according to  Eqn. (\ref{a15}) and (\ref{a38})  , for $0\leq k\leq 2$, we have 
\begin{equation}
\label{a24} \mathbb{E}[B] \leq   \frac{2cC\sqrt{R}}{\beta} \sqrt{|E_{R+1}|\sum_{r=1}^{R} \frac{3}{r^{\frac{k}{2}}}}  +  1. 
\end{equation}
By Eqn. (\ref{a21}) and Eqn. (\ref{a16}), Eqn. (\ref{a24}) becomes 
\begin{equation}
\mathbb{E}[B] \leq   \frac{2cC}{\beta} \sqrt{ \frac{6}{2-k}|E_{R+1}|}R^{1-\frac{k}{4}} +  1.
\end{equation}And for $k>2$, we have 
\begin{align}
\mathbb{E}[B] &\leq   \frac{2cC}{\beta} \sqrt{ 3|E_{R+1}|(\ln R+1) R}+  1\\
&=O(\sqrt{|E_{R+1}|(\ln R+1) R}). 
\end{align}Thus we complete the proof for Theorem \ref{friendship bound}. \end{proof}
Theorem \ref{friendship bound} implies that the regret bound of $\mathbb{EIM}$ is still sublinear to the number of trials under the growing network size. And the sub-linearity of the regret bound justifies that $\mathbb{EIM}$ can effectively capture the evolving network states with the bandits-based framework and achieve the long-run performance that converges to the optimal strategy. In Section \ref{experiments}, we will further experimentally demonstrate the performance of $\mathbb{EIM}$.

\section{Experiments}\label{experiments}
In this section, we experimentally evaluate the performance of $\mathbb{EIM}$ on both real world and synthetic evolving networks to investigate the following key issues. (1) Can the seeds set selected by $\mathbb{EIM}$ consistently outperform state-of-art methods in the {\bf EIM} problem? (2) Does the particle learning method capture network growing speed well? (3) Is the running time of $\mathbb{EIM}$ scale well in large scale networks? (4) What are the effects of seeds set size $K$ and the time on the performance of $\mathbb{EIM}$? To answer the four questions, we will first introduce the evolving network datasets constructed in our experiments and then provide the detailed settings and results. For space limitations, we only present partial representative results here, with more exhaustive results shown in Appendix \ref{Sexp}.  

  \begin{figure*}[t]
\subfigure[Growing nodes]
  {\begin{minipage}[b]{0.185\textwidth}
  \includegraphics[width=1\linewidth]{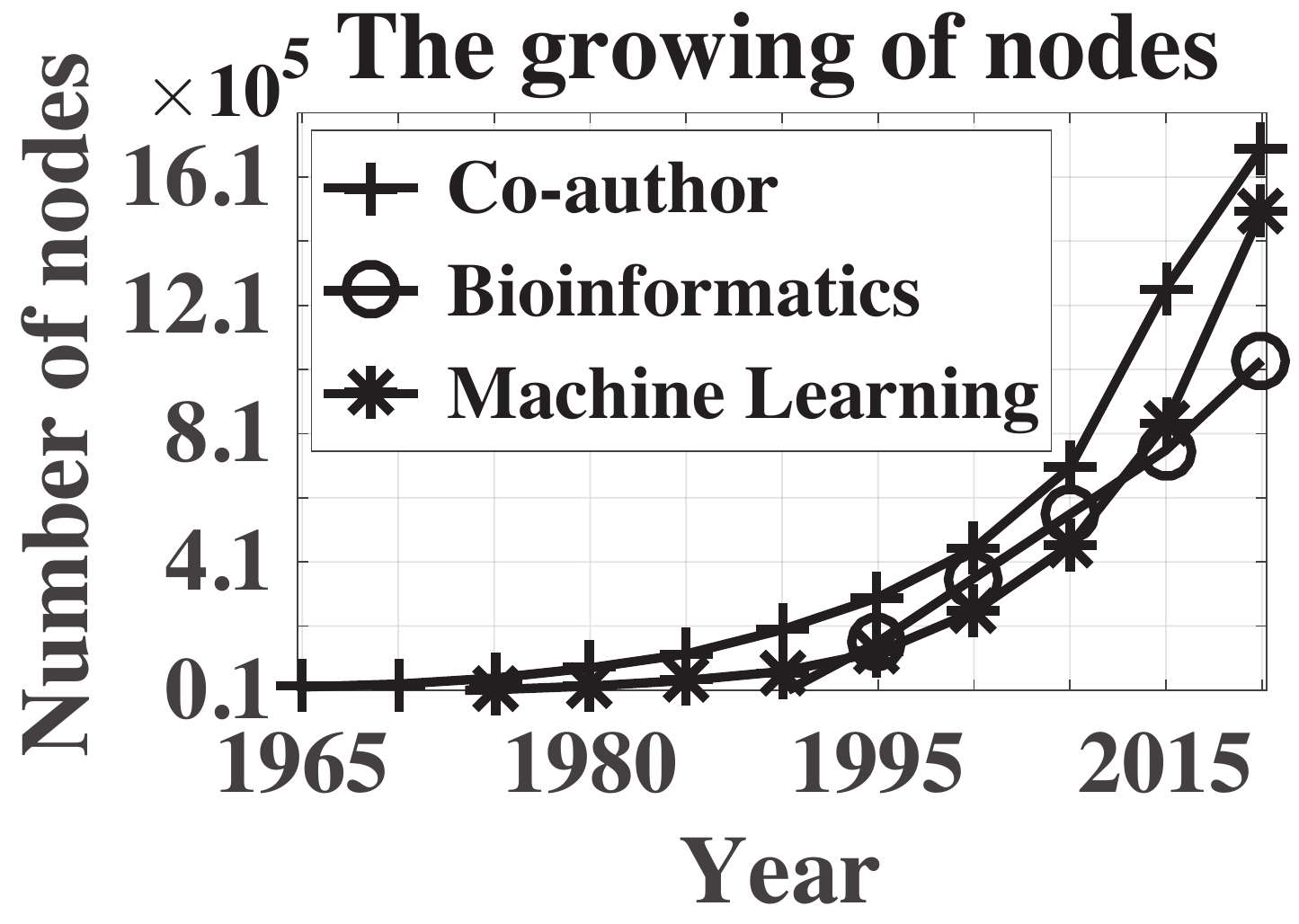}\label{G1}
  \end{minipage}}
  \hspace{-1mm}
  \subfigure[Growing degrees in Co-author]
  {\begin{minipage}[b]{0.19\textwidth}
  \includegraphics[width=1\linewidth]{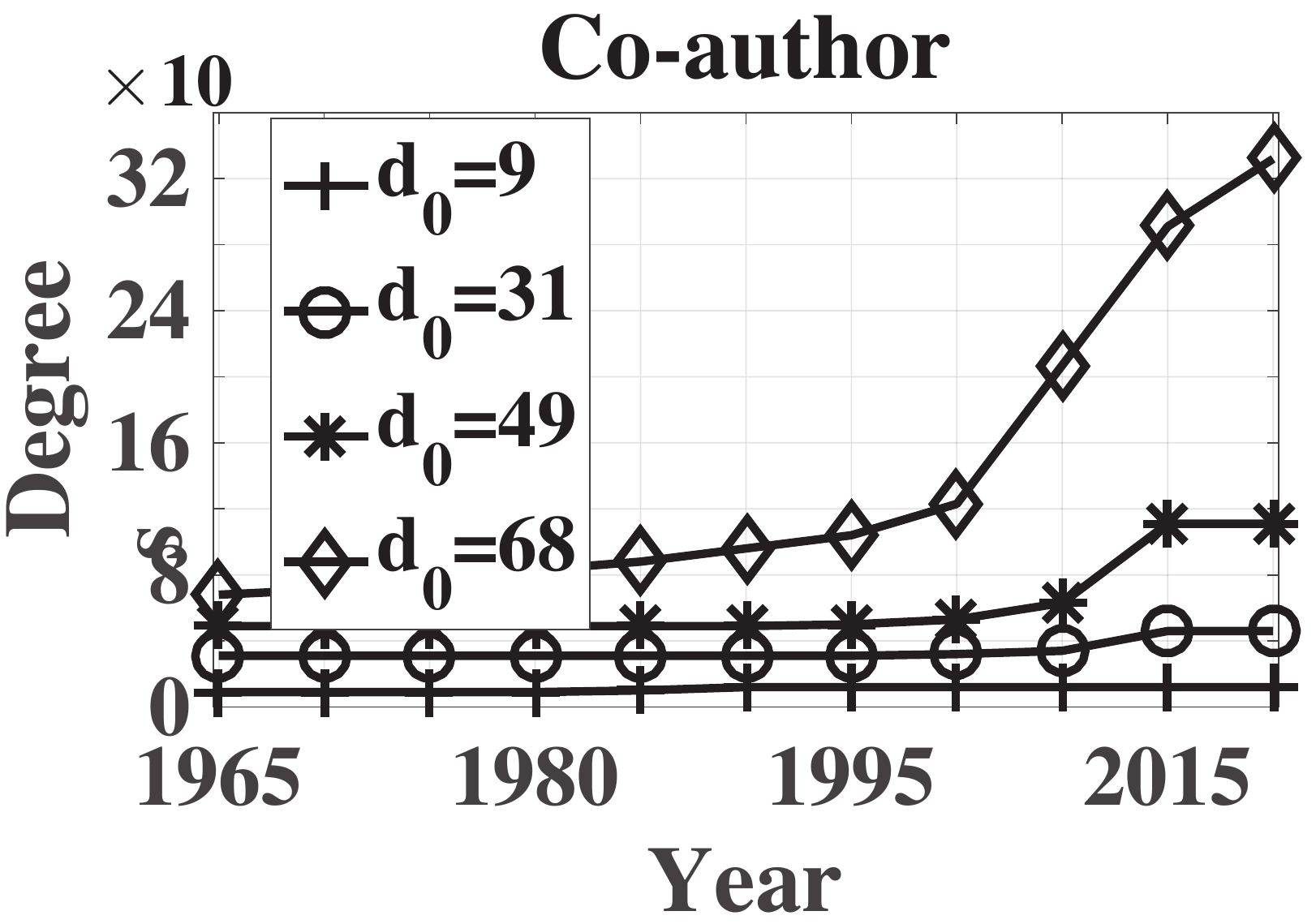}\label{G2}
  \end{minipage}}
    \hspace{-1mm}
  \subfigure[Growing degrees in Topic]
  {\begin{minipage}[b]{0.19\textwidth}
  \includegraphics[width=1\linewidth]{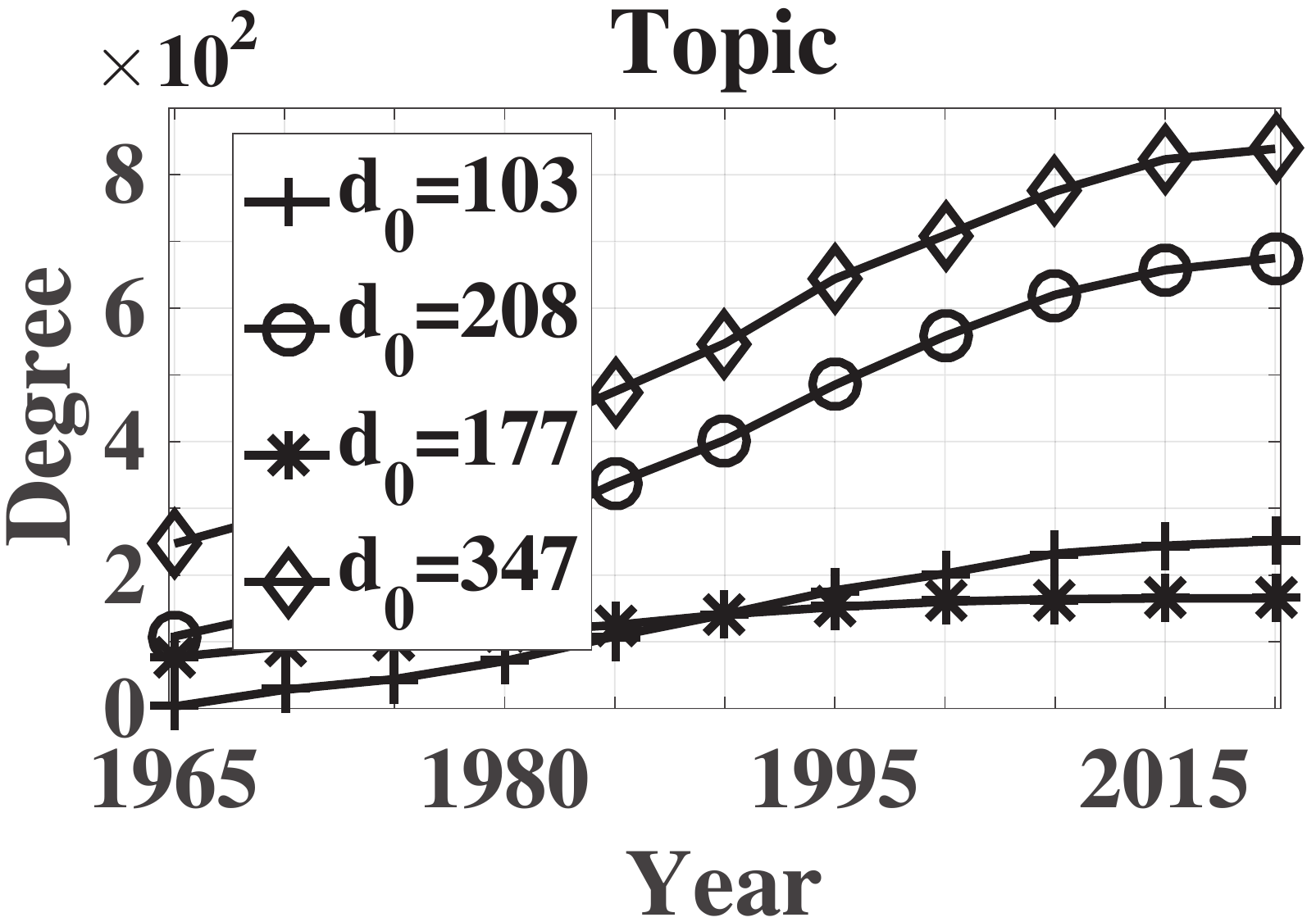}\label{G3}
  \end{minipage}}
    \hspace{-1mm}
   \subfigure[Growing degrees in Boinformatics]
  {\begin{minipage}[b]{0.19\textwidth}
  \includegraphics[width=1\linewidth]{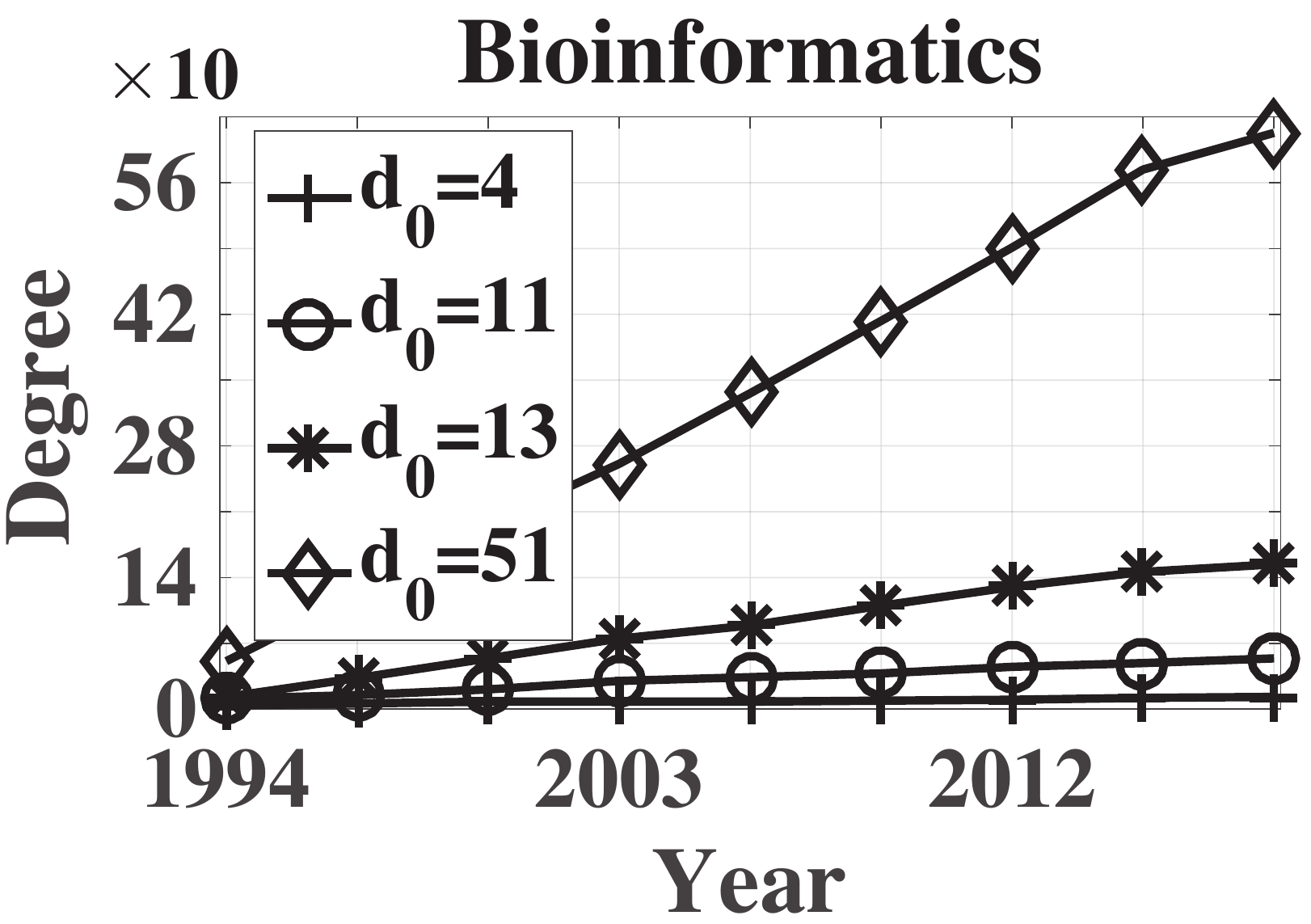}\label{G4}
  \end{minipage}}
    \hspace{-1mm}
    \subfigure[Growing degrees in ML]
  {\begin{minipage}[b]{0.195\textwidth}
  \includegraphics[width=1\linewidth]{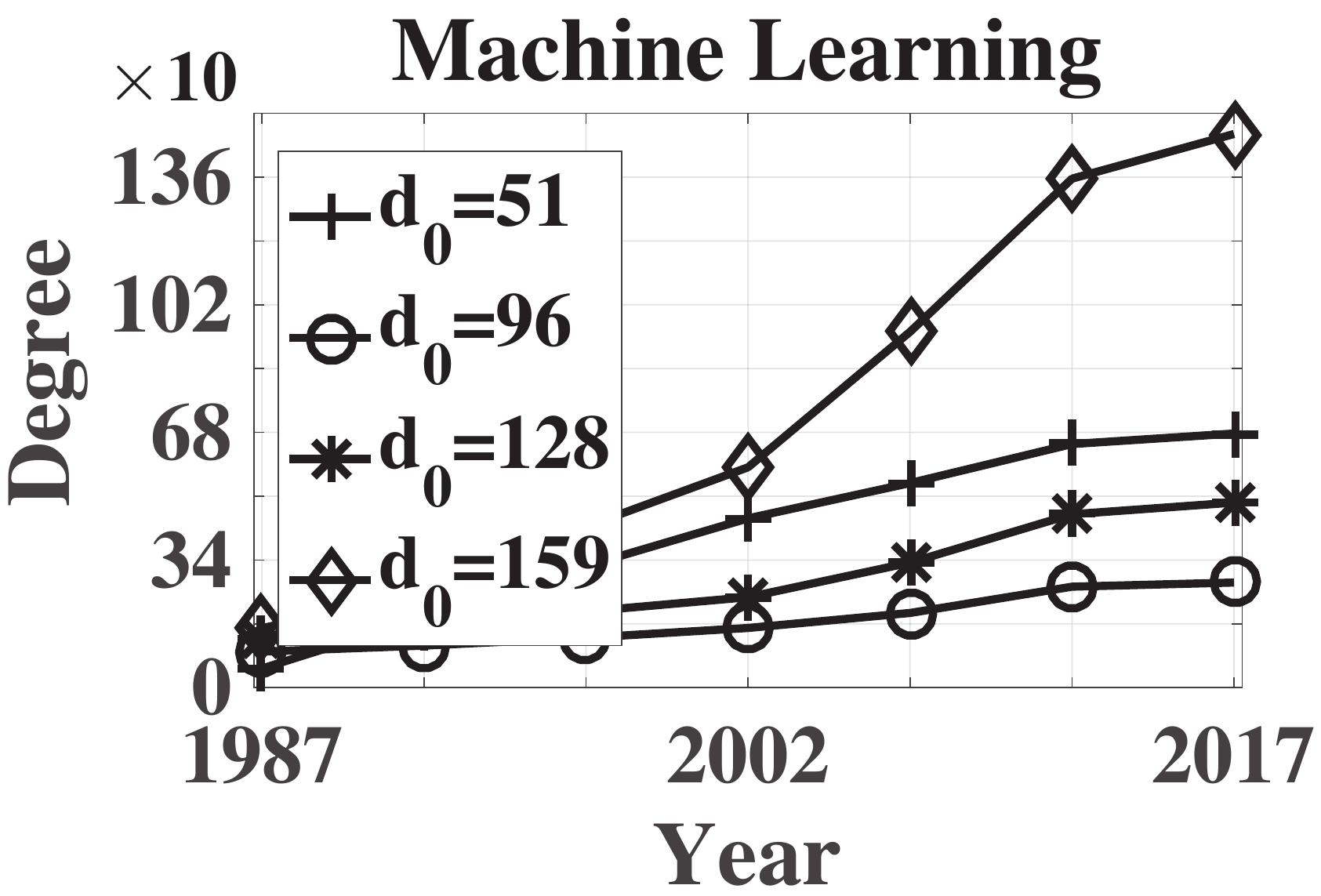}\label{G5}
  \end{minipage}}
    \hspace{-1mm}
  \vspace{-2mm}
 \caption{Evolution of networks. ($d_{0}$ means the initial degree)}\label{Evolutiongraph}
 \vspace{-2mm}
\end{figure*}
\vspace{-2mm}
\subsection{Evolving Network Datasets}\label{datasets}
Since existing widely used social network datasets lack complete information of the joining time of each node, we extract four real evolving networks from the Microsoft Academic Graph (MAG) \cite{co-author}. Besides, we also generate a synthetic network following the Barabasi-Albert evolving model \cite{barabasi1999emergence}. The statistical details of the five datasets are summarized in Table \ref{data}.

  \begin{table}[h]

    \vspace{-3mm}
\centering
\caption{Statistics of Evolving Datasets}\label{data}
\label{dataset}
\begin{tabular}{c|c|c|c} \hline
{\bfseries Datasets} &{\bfseries \# of Nodes} &{\bfseries \# of Edges}  \tnote{1} &{\bfseries Time Interval}
\\ \hline
Co-author&1.7M&12.6M&A.D. 1801-2015\\
Topic &34K&727K &A.D. 1800-2016 \\
ML &1.51M&6.9M &A.D. 1872-2017 \\
Bio &1.04M&1.82M & A.D. 1992-2017 \\
SN&420 K &3.86M& 25 periods \\\hline
\end{tabular}

\vspace{-4mm}
\end{table} 

{\bfseries (1) Co-author:} From the author list of each paper, we extract a co-authorship evolving network which contains $1.7$ million nodes and $12.6$ million edges. The edge between a pair of authors means there are at least one paper co-authored by them. The joining time of each user is set to the publishing time of his first paper. When an author joins the network in evolution as time goes on, we connect him with his co-authors who have already joined in.

{\bfseries (2) Topic:} There are $127$ million papers in the MAG dataset, and we classify them into $34 K$ topics with reliable ground-truth communities. The joining time of each topic is set to the publishing time of its earliest paper, and edges  are based on the citations among topics. We say that topic $1$ cites topic $2$ if a paper belonging to topic $1$ cites another paper belonging to topic $2$.  Since the cross-domain citations are widely existed in academia, Topic is the densest one in the five networks. Table \ref{rtopic} lists the statistics of several representative topics. 

  \begin{table}[h]
    \vspace{-3mm}
\centering
\caption{Statistics of Evolving Datasets}
    \vspace{-2mm}
\label{rtopic}
{\footnotesize
\begin{tabular}{c|c|c|c} \hline
\multirowcell{2}{{\bfseries Topic}} &\multirowcell{2}{{\bfseries \# of Papers}} &{\bfseries Joining} &\multirowcell{2}{{\bfseries  First Paper }}\\
&&{\bfseries Time}&
\\ \hline
Computer network &380 K& 1850& No place like home\\ \hline
Computer  vision &1.2 M &1879 & Survival of the Fittest \\\hline
\multirowcell{2}{World wide web} &\multirowcell{2}{349 K}&\multirowcell{2}{1848}&\multirowcell{2}{The past, the present, \\and the future}  \\
&&&\\\hline
\end{tabular}}
\vspace{-3mm}
\end{table}

{\bfseries (3) Machine Learning (ML):} The evolving ML network is composed of the papers belonging to the Machine Learning topic, which contains $1.51$ million nodes and $6.9$ million edges. The joining time of each node is set as the publishing time of its corresponding paper, and the edges are established based on citations among papers.

{\bfseries (4) Bioinformatics (Bio):} The evolving Bio network includes $1.04$ million papers about the Bioinformatics topic, which contains $1.04$ million nodes and $1.82$ million edges. Its construction method is similar to that of ML network. 

{\bfseries (5) Synthetic Network (SN):} We also generate a synthetic network that includes $420 K$ nodes and $3.86M
$ edges based on the Barabasi-Albert (BA) evolving model \cite{barabasi1999emergence}. In the generation, a new node is attached to the previous graph by a single edge in each evolving time slot. With probability $\frac{1}{2}$, the anchor node is chosen uniformly at random from nodes in previous graph. Otherwise, the possibility of an anchor node being selected is proportional to its current degree. 

Figure \ref{Evolutiongraph} plots the growth of nodes and degrees in the five evolving networks. From Figure \ref{G1}, we can find that Co-author, ML and Bio all follow the power-law growth. And from Figures \ref{G2}-\ref{G5}, we can see that the gap between the degrees of users grows with network evolution.  Especially, the node with highest initial degree exhibit significant advantage in later years. The phenomena well justify the BA evolving model where new users will preferentially connect to those with higher degrees in evolution.
 
  \begin{table*}[t]
\centering
\caption{Influenced size over $K$ in $2015$ }\label{eff}
\resizebox{\textwidth}{!}{ 
\begin{tabular}{c|cccc|cccc|cccc|cccc|cccc} \hline
&\multicolumn{4}{|c|}{Co-author} &\multicolumn{4}{|c|}{ML} &\multicolumn{4}{|c|}{SN} &\multicolumn{4}{|c|}{Topic}&\multicolumn{4}{|c}{Bio}
\\ \hline
Algorithm&K=5&10&20&50&5&10&20&50&5&10&20&50&5&10&20&50&5&10&20&50
\\ \hline
IMM&39k&40k&45k&47k&14k&26k&58k&106k&203	&927&1.9k&3.8k&3k	&4.5k	&7.7k&	10k&602&	1.2k&1.6k&2.6k
\\
$\mathbb {EIM}$&41k&46k&	55k&	72k&32k&	46k&	68k&	117k&816&1.3k	&2.3k&4.6k&2.9k&6.3k&7.7k&	14k&981	&1.5k&	2k&	3.3k
\\
HD&36k&	37k&	46k&	46k&17k&	34k&	63k&	92k&168&	1k&2.9k&	4k&2.7k&2.7k&	7.2k&13k&473&1k&	1.5k&	2.7k
\\ 
Earliest&5	&10&8.9k&46k&20&	57&	88	&689&170	&888	&2.4k&2.9k&2.7&6.1k&7.7k&	13k&5&10	&23&67
\\ 
SKIM&37k&37k&44k&46k&25k&32k&	54k&	103k&249	&940&1.7k&2.8k&1.9k&	5.2k&7.3k&	12k&486&979	&1.6k&	2.4k
\\ \hline
\end{tabular}}
\end{table*} 
     \begin{figure*}[t]
 \centering
 \vspace{-1mm}
\centering
  \includegraphics[width=1\textwidth]{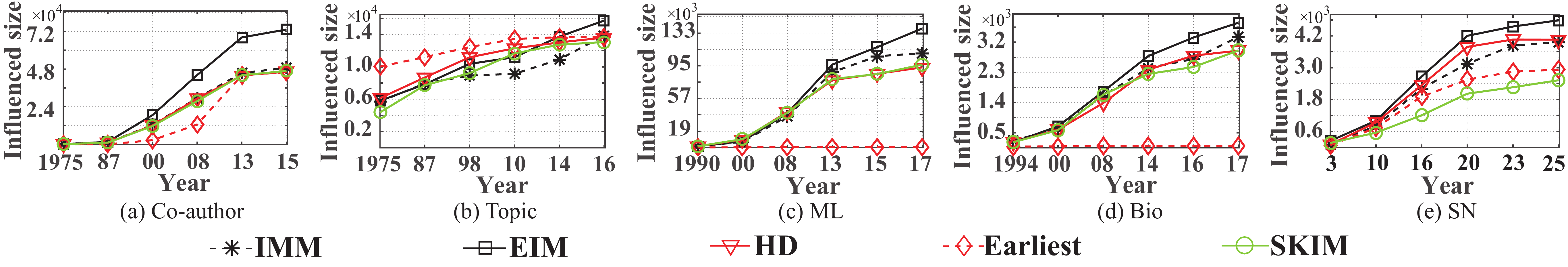}
  \vspace{-9mm}
    \caption{Influenced size over years under $K=50$. }\label{k50year}
    \vspace{-4mm}
    \end{figure*}
     \vspace{-2mm}
\subsection{Experimental Settings}

 \vspace{-1mm}
{\bfseries Baselines}. We compare the performance of $\mathbb{EIM}$ with the following four baseline algorithms:

(1) IMM \cite{IMM}: The general IMM framework focuses on efficiently selecting seed users for IM problem over large scale static networks. Its main idea lies on determining the number of RR-sets to ensure the approximation ratio, and then iteratively selecting seed users who can cover the most RR-sets.

(2) SKIM \cite{cikm1}: It adopts a construction of reachability sketches of static networks under IC model. Then a reverse reachability search over such sketches is performed to iteratively select seed users  who can firstly reach $k$ nodes through active edges.

(3) Highest Degree (HD): A heuristic algorithm that selects $K$ seed users with highest degrees in each trial.  

(4) Earliest: A heuristic algorithm that selects $K$ seed users who join the evolving network earliest in each trial. 

Note that the two static baselines, i.e., IMM \cite{IMM} and SKIM \cite{cikm1} cannot be applied to evolving networks directly, thus their settings in the experiments are not exactly the same with their originalities. For IMM and SKIM, seeds are still selected from nodes that have ever been observed in each trial instead of the entire network due to partial observing.
%

{\bfseries Parameter settings.} We set one year as the period for each trial, and the time for the first trial in the five datasets are set as  (1) Co-author: $1965$, (2) Topic:  $1961$, (3) ML:  $1988$, (4) Bio: $1993$ and (5) SN: the first trial is conducted at the timestamp when network size is $2500$, and the size of new users in the $n$-th period is set as $2500\cdot 2^{n}$. The initial weights of edges are sampled from $\mathcal{N}(0.05, 0.008)$.
Regarding the unidirected citation patterns in Topic, ML and Bio, we set the weight $w_{AB}$ of edge where $A$ cites $B$ as a Gaussian random walk and let $w_{BA}$ always be $0$. In bidirected networks Coauthor and SN, the weights of two directions are both set as the Gaussian random walk, while the two weights are independent. The default number $M$ of particles is set to $500$, and effect of  $M$ will be shown later in Section \ref{performance of particle}. The default value of $\varepsilon$ in Sampling phase (Algorithm \ref{Sampling}) is set as $\varepsilon=0.1$, whose effect further graphically reported in Section \ref{EA}

 {\bfseries Environment.} All the experiments are implemented in Python 2.7 and conducted on a computer running Ubuntu 16.04 LTS with 40 cores 2.30 GHz (Intel Xeon E5-2650) and 126 GB memory.
 \vspace{-4mm}
 \subsection{Effectiveness of $\mathbb{EIM}$}
We quantify the effectiveness by the number of influenced users and report the comparison of the effectiveness between $\mathbb{EIM}$ and four baselines in Table \ref{eff} and Figure \ref{k50year}.
 
 {\bf Effects of Time}. From Figure \ref{k50year}, we can observe that over the Co-author, ML, Bio and SN, $\mathbb{EIM}$ always outperform the four baselines. And the superiority of $\mathbb{EIM}$ becomes more significant as time increases. Especially in the case that $K=50$, $Year= 2015$ over Co-author, the influenced size of $\mathbb{EIM}$ is almost $50\%$ larger than that of baselines. The superiority of $\mathbb{EIM}$ owes to the continuous learning of network knowledge, so that with more accurate network knowledge,  $\mathbb{EIM}$ can return better seeds set. This phenomena justifies that the IM designing and network knowledge learning can mutually enhance each other.

Over Topic,  it can be seen that the influenced size of $\mathbb{EIM}$ is smaller than that of HD in early years, since the uncertainties of network knowledge degrade the performance of $\mathbb{EIM}$ as well as IMM and SKIM. Specifically, Topic is the densest network where each new node (topic) averagely cites more than 20 existing nodes, so that there is a higher probability that it cites the $50$ highest degree nodes under the PA rule.
However, even over the special case, $\mathbb{EIM}$ still enjoys better performance than the four baselines in later years.

{\bf Effects of $K$.} From Table \ref{eff},  we can find that the influenced sizes of IMM and SKIM grow smoothly with the increase of $K$, while those for Earliest presents much more fluctuations. The reason behind is that IMM and SKIM are efficient IM algorithms over static networks with rigorous performance guarantee, their disadvantages to $\mathbb{EIM}$ is brought by the inapplicability in evolving framework. Meanwhile, Earliest is a  heuristic with no performance guarantee, and the instability of its performance implies the heterogeneity of user attractiveness, especially among those join in early stage. In contrast, another heuristic HD achieves medium influenced size among the five algorithms, since HD benefits from the PA rule.

   \vspace{-2mm}
  \subsection{Efficiency of $\mathbb{EIM}$}
   \vspace{-1mm}
   Now, we report the running time of $\mathbb{EIM}$ in Figure \ref{efficiency}.  As shown in Theorem \ref{ratio}, the evolving IM algorithm {\bfseries Evo-IMM} costs $O\big((K+l)((|\mathbb{V}^{r}|+|\mathbb{E}^{r}|)+\frac{|\mathbb{V}^{r}|}{OPT})\log |\mathbb{V}^{r}|/\varepsilon^{2}\big)$. Another phase of {\bfseries Evo-IMM} is the network knowledge learning, whose computational complexity can be scaled as $O(M|\mathbb{V}^{r}|+|\mathbb{E}^{r}|)$. Thus the running time of $\mathbb{EIM}$ is proportional to the network size as shown in Figure \ref{efficiency}. Due to the high efficiency of linear UCB and IMM frameworks, the time costs of $\mathbb{EIM}$ scales well even over networks of million scale. The running time of several classical IM algorithms over million-scale networks: TIM $(10^4 s) \cite{fIMM} $ , TIM+ $(10^3 s) \cite{fIMM}$, IMM $(10^2-10^3 s)$ \cite{IMM}, respectively. Another two classical algorithms (i.e., RIS and CELF++), according to the experimental results in \cite{fIMM}, cost $10^4$ seconds over the network with $76K$ nodes. What we can also find from Figure \ref{efficiency} is that the increase of $K$ from $10$ to $20$ only incurs slightly larger  time costs, ensuring the scalability of $\mathbb{EIM}$ in the cases where a large number of seeds need to be selected.
   \begin{figure}[h]
  \vspace{-3mm}
 \centering
\centering
  \includegraphics[width=0.48\textwidth]{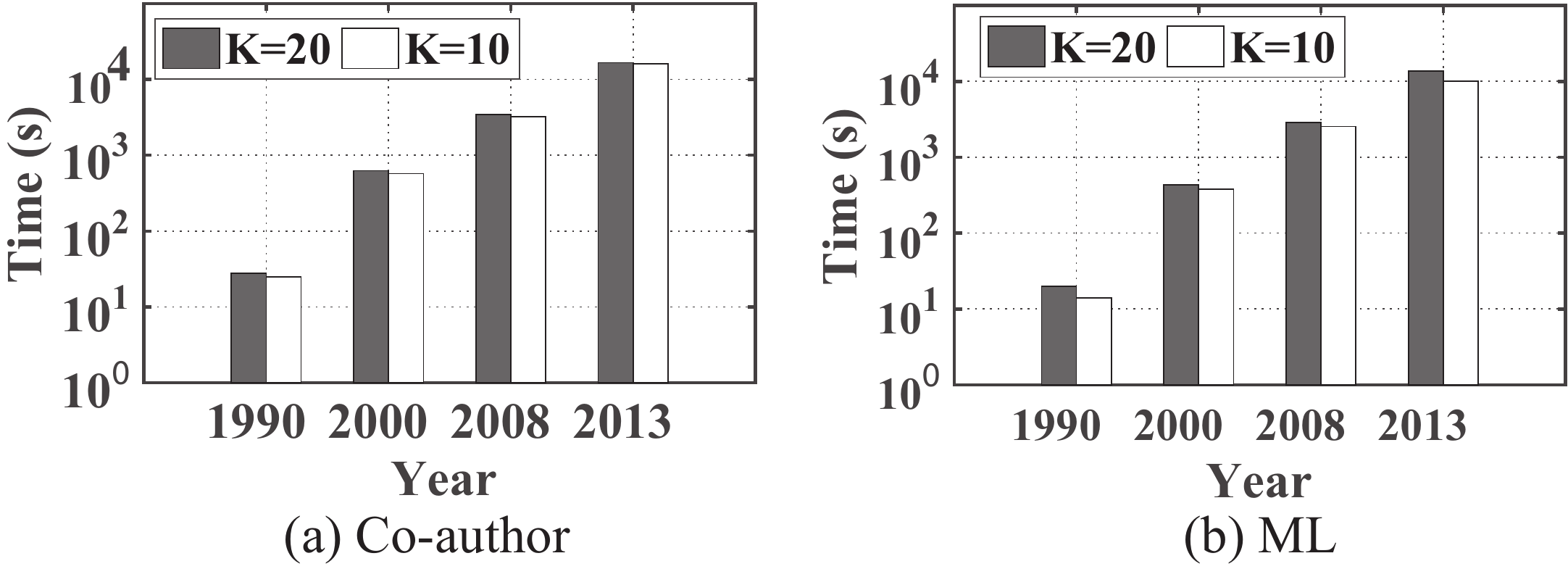}
  \vspace{-4mm}
  \caption{Running time vs. Year over Co-author and ML}\label{efficiency}
   \vspace{-6mm}
  \end{figure}
   \subsection{Performance of Particle Learning }\label{performance of particle}
     \vspace{-1mm}


Let $M$ denote the number of initial sampled particles. Referring \cite{sigmoids}, the initial prior parameters are sampled from their possible ranges as:  $\beta \in [10^{-8}, 1]$, $\theta \in [10^{-4}, 10]$, $N \in [10^{5}, 10^{8}]$. We define the metric, i.e., the relative error $|(\sum_{p_{i}\in \mathcal{P}^{r}}n_{i}(T^{r}))/M-n(T^{r})|/n(T^{r})$ to measure the accuracy of learnt network growing speed. Figure \ref{particleM} plots the relative errors over Topic and SN with initial size being $M=500$ and $1000$. It can be seen that  the accuracy of particle learning can be improved by the size of initial particles as more particles bring higher resolution of initial parameters. Also, the accuracy increases over time, since particles with accurate parameters are gradually filtered out in the resampling phases of each trial. \begin{figure}[h]
  \vspace{-2mm}
 \centering
\centering
  \includegraphics[width=0.48\textwidth]{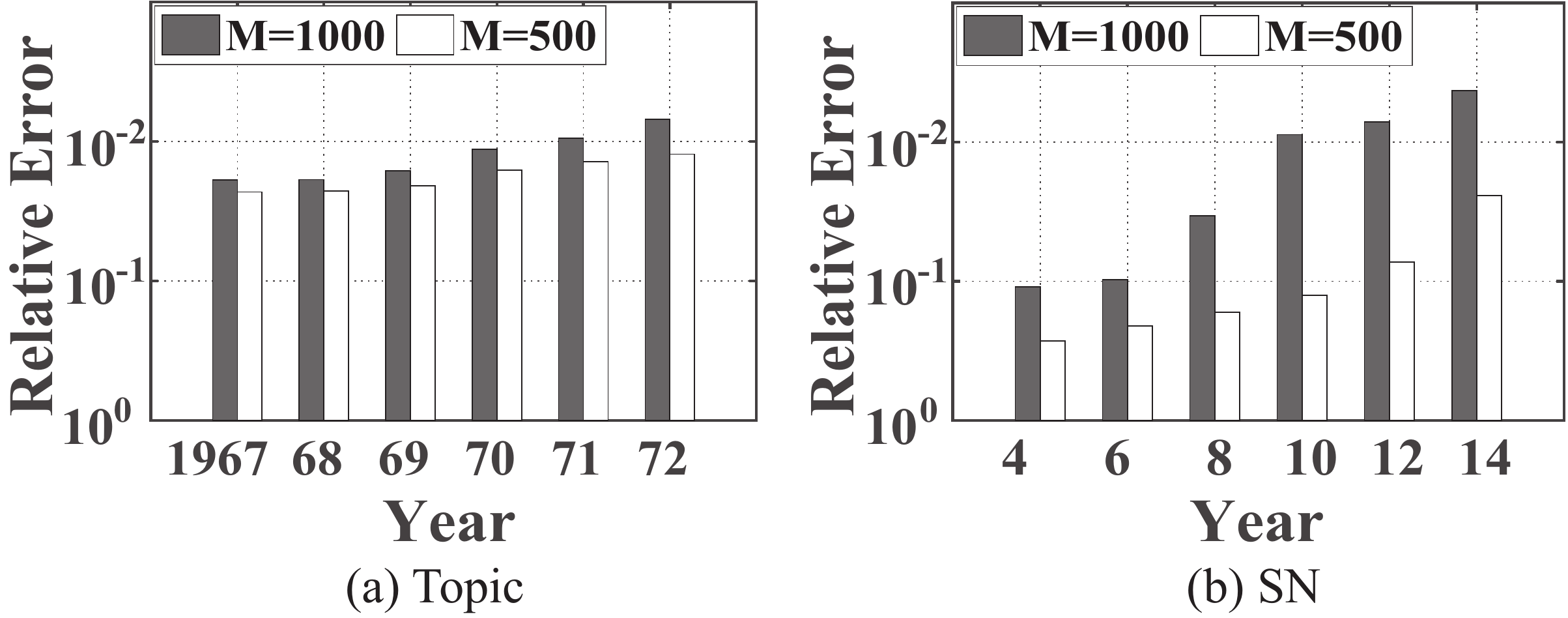}
  \vspace{-4mm}
  \caption{Relative error vs. M over Topic and SN}\label{particleM}
   \vspace{-6mm}
  \end{figure}   
  
  \subsection{Effects of Parameter $\varepsilon$ }\label{EA}
 
 Recall that in Theorem \ref{ratio}, the approximation ratio and computational complexity are both the functions of parameter $\varepsilon$. Figure \ref{EB} shows the effect of $\varepsilon$ over ML with $K=20$. To intuitively illustrate the effect of  $\varepsilon$ on time costs of seeds selection, we present the running time of {\bfseries Evo-IMM} with $\varepsilon=0.1$ and $\varepsilon=0.5$ in Figure \ref{EB}. Since lager $\varepsilon$ means smaller number of RRsets needed in {\bfseries Evo-IMM}, the running time of cases where  $\varepsilon=0.5$ is much smaller than that of the cases where $\varepsilon=0.1$. Although the increase of $\varepsilon$ causes the decrease  theoretical performance guarantee,  in the experiment {\bfseries Evo-IMM} achieves comparable expected influenced size when $\varepsilon=0.5$. 
 
     \begin{figure}[h]
 \centering
\centering
  \includegraphics[width=0.48\textwidth]{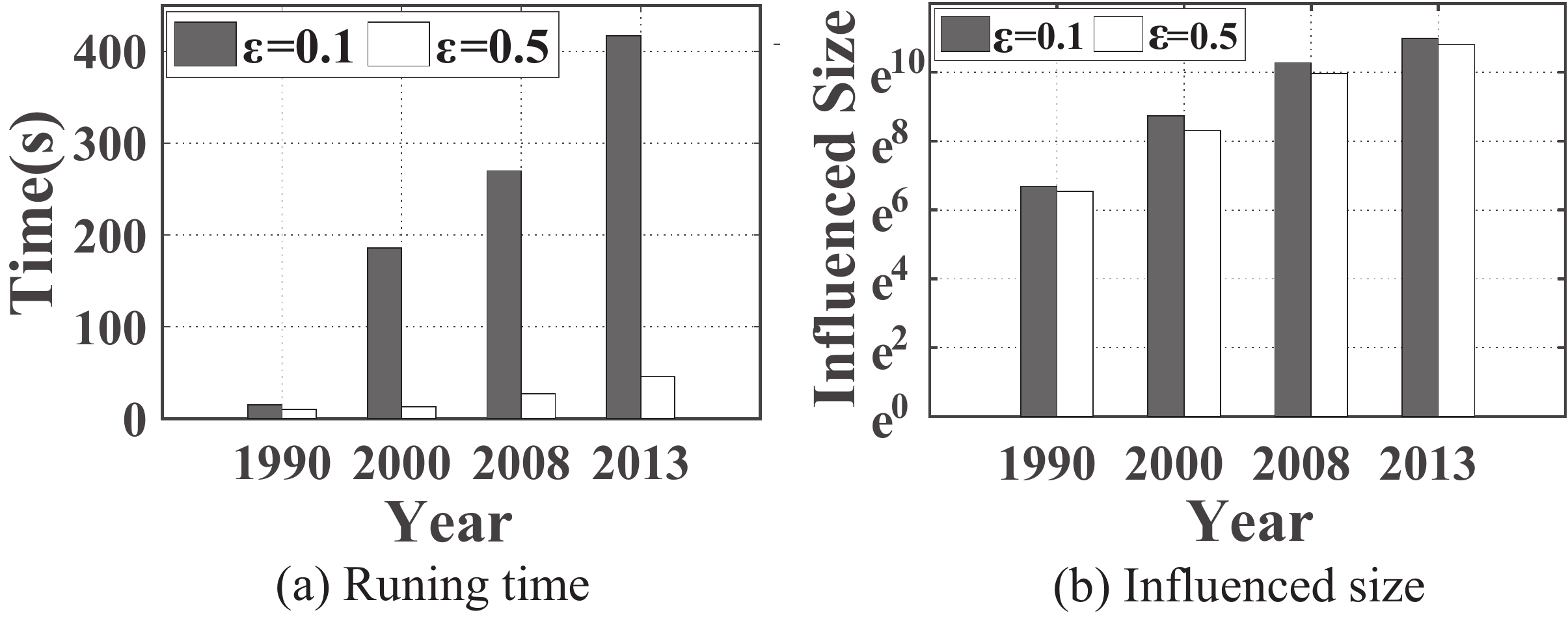}
  \vspace{-4mm}
  \caption{Effect of $\varepsilon$ over ML}\label{EB}
  \end{figure}

%

\section{Conclusion}
This paper investigates the influence maximization in evolving networks where new users continuously join with influence diffusion. A bandits based framework $\mathbb{EIM}$ is proposed to simultaneously design IM and learn network knowledges. In each trial, a particle learning method is first adopted to learn the network growing speed based on the preferential attachment rule. And an UCB based framework is designed to learn evolving influences among users. Under the refined growing speed and influences, we propose an evolving IM algorithm {\bfseries Evo-IMM} to efficiently select the seed users for evolving IM. We show that the regret bound of $\mathbb{EIM}$ is sublinear to the number of trials. At last, the experiments on both real and synthetic evolving network datasets demonstrate that {\bfseries EIM} outperforms four baselines in solving {\bf EIM} problem.

\bibliographystyle{unsrt}

 \appendices
 \section{Proof for Lemma 1}\label{hardnessproof1}
 
Lemma \ref{hardness}. The {\bfseries EIM} problem is {\emph NP-hard}. The computation of $I(S, G_{t+T})$ is {\emph \#P-hard}. And the objective function $I(S, G_{t+T})$ is monotone and submodular.

\begin{proof}
{\bf NP-hardness.} To prove the NP-hardness, we reduce the problem proposed in Eqn. (\ref{problem}) to the `set cover' problem described as follows. Give a universe $\mathcal{U}=\{x_{1}, x_{2}, ..., x_{n}\}$ and a collection $\mathcal{C}$ of subsets $\mathcal{C}=\{C_{1}, C_{2},..., C_{n'}\}$, the goal of `set cover' is to find a cover $\mathcal{A}\subseteq \mathcal{C}$ with size $K$ whose union equals to the universe $\mathcal{U}$. Then the `set cover' is reduced to the evolving IM problem as follows. We construct a corresponding bipartite graph $G$ that consists of the subset partition and the element partition. In subset partition, there are $n'$ nodes representing the subsets in collection $\mathcal{C}$. And the element partition consists of $n$ nodes representing the elements in $\mathcal{U}$. If element $x_{i}\in C_{j}$, there is an edge with the weight being $1$ from node $C_{j}$ to node $x_{i}$ in $G$. Then the `set cover' problem is equivalent to deciding whether there is a set of $K$ nodes in $G$ with the influence being $K+n$. Since the `set cover' problem is NP-hard, the evolving IM problem is NP-hard. 

{\bf \#P-hardness.} To prove the \#P-hardness, we reduce computing $I(S, G_{t+T})$ from the {\emph S-D connectivity} counting problem described as follows. Given a graph $G=(V,E)$ and a pair of Source (S) and Destination (D) nodes, the {\emph S-D connectivity} problem is to compute the probability that S and D are connected given each edge in $G$ has an independent probability of $0.5$ to be connected. We reduce the  {\emph S-D connectivity} to computing $I(S, G_{t+T})$ as follows. 
 Assuming that the edges in $G_{t+T}$ has an independent probability of $p=0.5$ to be connected, computing $I(S, G_{t+T})$ is equivalent to counting the expected number of nodes that are connected to the nodes belonging to $S$. Since the {\emph S-D connectivity} counting problem is \#P-hard, the computation of $I(S, G_{t+T})$ is \#P-hard.

{\bf Monotonicity.} We consider an instance of $G^{r+1}$, i.e., $\overline{G}^{r+1}$ where the state of each edge $e$ is determined by flipping a coin of bias $w_{e,r+1}$. If the coin representing edge from node $u$ to node $v$ flips, user $u$ can successfully influence user $v$ after he has been influenced. Let $S_{1}$ and $S_{2}$ denote two seed sets with $S_{1}\subseteq S_{2}$, and $I(S_{1})$ and $I(S_{2})$ respectively denote the nodes influenced by $S_{1}$ and $S_{2}$ over $\overline{G}^{r+1}$. For any node $a$ in $I(S_{1})$, since $S_{1}\subseteq S_{2}$, there must be an active path from a node in $S_{2}$ to $a$ . Thus we have  $I(S_{1})\subseteq I(S_{2})$, which demonstrates the monotonicity of influence function in Eqn. (\ref{EIM problem}). 

 {\bf Submodularity.} Furthermore, let $S_{3}=S_{1} \cup x$, $S_{4}=S_{2} \cup x$, $I(S_{3}\backslash S_{1})=I(S_{3})\backslash I(S_{1})$ and $I(S_{4}\backslash S_{2})=I(S_{4})\backslash I(S_{2})$. For a node $a\in I(S_{4}\backslash S_{2})$, there is an active path from $x$ to $a$ while no active path from $S_{1}$ to $a$. Since $S_{1}\subseteq S_{2}$, $S_{1}$ cannot influence $a$ over $\overline{G}^{r+1}$. Thus $a \in I(S_{3}\backslash S_{1})$ and $I(S_{4}\backslash S_{2}) \subseteq I(S_{3}\backslash S_{1})$, which demonstrates the submodularity of influence function in Eqn. (\ref{EIM problem}). \end{proof}

 \section{Proof for Lemma 2}\label{Lemma 4.1}

Lemma 2. Given the degree of node $v_{n}$ at time $t$ is $d_{n}^{t}$ and the period $T$ of each trial,  we have
\begin{equation}
\notag \mathbb{E}(d_{n}^{T+t})=d_{n}^{t}\cdot \prod_{s=1}^{m[n(t+T)-n(t)]} \left(1+\frac{1}{\sum_{v_{j}\in V_{t}}d_{j}^{t}+(2s-1)}\right).
\end{equation}
\begin{proof}
We first consider a special case when $m=1$. According to Eqn. (\ref{PR}), at each evolving time slot, we have
\begin{equation}
\label{degree1}\mathbb{E}(d_{n}^{t+\Delta t})=d_{n}^{t}\cdot \left(1+\frac{1}{\sum_{v_{j}\in V_{t}}d_{j}^{l}+1}\right).
\end{equation}Since a new edge establishes in time slot $t+\Delta t$, the total degrees of all nodes after time slot $t+\Delta t$ becomes $\sum_{v_{j}\in V_{t}}d_{j}^{l}+2$. Then the expected degree of node $v_{n}$ at time slot  $t+2\Delta t$ is
\begin{align}
\label{degree2}\mathbb{E}(d_{n}^{t+2\Delta t})&=\mathbb{E}(d_{n}^{t+\Delta t})\cdot\left(1+\frac{1}{\sum_{v_{j}\in V_{t}}d_{j}^{l}+3}\right)\\
\vspace{-1mm}
\label{degree3}&=d_{n}^{t}\cdot \prod_{s=1}^{2} \left(1+\frac{1}{\sum_{v_{j}\in V_{t}}d_{j}^{t}+(2s-1)}\right).
\vspace{-1mm}
\end{align}Here, $\Delta t$ denotes an evolving slot. Under the growing speed $n(t)$, there are $n(T+t)-n(t)$ new nodes joining the network during $t$ to $t+T$. Thus there are $n(T+t)-n(t)$ evolving time slots during $t$ to $t+T$, and we have
\begin{equation}\label{degree4}
 \mathbb{E}(d_{n}^{T+t})=d_{n}^{t}\cdot \prod_{s=1}^{n(t+T)-n(t)} \left(1+\frac{1}{\sum_{v_{j}\in V_{t}}d_{j}^{t}+(2s-1)}\right).
\end{equation}

Then we consider the general cases when $m\geq 2$. Under the PA rule, the $m$ new edges brought by a same new node are respectively established in $m$ evolving time slots. Thus there are $m[n(T+t)-n(t)]$ evolving time slots during $t$ to $t+T$ in the general cases. Then  Eqn. (\ref{degree4}) inductively becomes 
\begin{equation}
\notag \mathbb{E}(d_{n}^{T+t})=d_{n}^{t}\cdot \prod_{s=1}^{m[n(t+T)-n(t)]} \left(1+\frac{1}{\sum_{v_{j}\in V_{t}}d_{j}^{t}+(2s-1)}\right).
\end{equation} Thus we complete the proof for Lemma \ref{degree}. 
\end{proof}
 \section{Derivations for \textnormal{$\mathbb{E}_{i}(d_{e}^{r+1})$}}\label{derivations for added degrees}
In each trial, we first take the growing function $n_{i}(t)$ into Lemma \ref{degree} to compute the expected incremental degrees of the observed nodes.  Such incremental degrees are then taken as the prior value of particle $\rho_{i}$. Let the real incremental degrees of observed nodes serve as the ground truth. Then a resampling process is conducted to resample the particles whose prior values are near the ground truth as more new particles,  while killing those with large deviations from the ground truth.

\begin{figure}[h]
 \centering
\centering
  \includegraphics[width=0.48\textwidth]{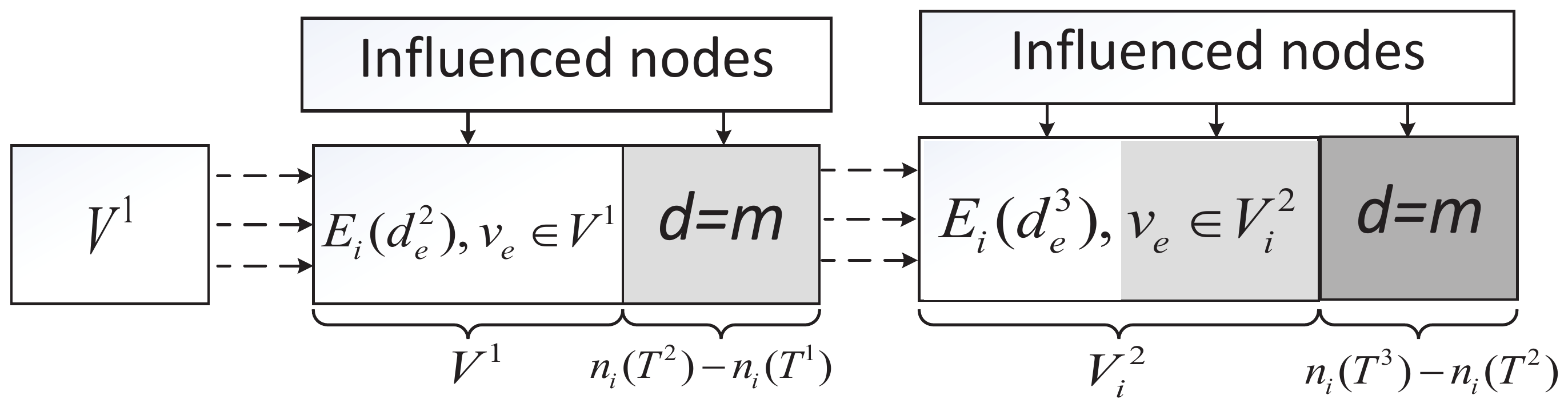}
  \vspace{-4mm}
  \caption{A sketch of evolving process under particle $\rho_{i}$ in the first three trials.}\label{evosketch}
   \vspace{-2mm}
  \end{figure}
  
{\bfseries Evolving process under particle $\rho_{i}$.} The evolving process starts from the given initial nodes set at the beginning of the first trial $V^{1}$, which is the same for all the particles. 
A sketch that contains the first three trials is shown in Figure \ref{evosketch}. Let $V_{i}^{r}$ denote the evolving nodes set under particle $\rho_{i}$ until time $T^{r}$ with $|V_{i}^{r}|=n_{i}(T^{r})$. In the first trial,  given the initial degree of a node $v_{e} \in V^{1}$ being $d_{e}^{1}$, according to Lemma \ref{degree}, its expected degree until time $T^{2}$ can be estimated as
{\small \begin{equation}\label{expected degree1}
\notag \mathbb{E}_{i}(d_{e}^{2})=d_{e}^{1}\cdot \prod_{s=1}^{m[n_{i}(T^{2})-n_{i}(T^{1})]}\left(1+\frac{1}{\sum_{v_{j}\in V^{1}}d_{j}^{1}+(2s-1)}\right).
\end{equation}}Besides, from time $T^{1}$ to $T^{2}$, there are $n_{i}(T^{2})-n_{i}(T^{1})$ newly added nodes in expectation under particle $\rho_{i}$, and the degrees of such nodes in $V_{i}^{2}\backslash V^{1}$ are uniformly expected as $m$, as shown in the middle part of Figure \ref{evosketch}. Here, $|V_{i}^{2}\backslash V^{1}|=n_{i}(T^{2})-n_{i}(T^{1})$. Then in the $2$-nd trial, given influenced nodes set during $T^{1}$ to $T^{2}$ being $O(T^{1})$, the degrees of nodes in $O(T^{1})$ are updated by their real observed degrees, while the others in $V_{i}^{2}\backslash O(T^{2})$ still reserve their estimating degrees. Thus for $v_{e}\in V_{i}^{2}$, its expected degree until $T^{3}$ equals to 
\vspace{-1mm}
{\small \begin{equation}\label{expected degree}
\notag \mathbb{E}_{i}(d_{e}^{3})=\widetilde{d}_{e}^{2}\cdot \prod_{s=1}^{m[n_{i}(T^{3})-n_{i}(T^{2})]} \left(1+\frac{1}{\sum_{v_{j}\in V_{i}^{2}}\widetilde{d}_{j}^{\,2}+(2s-1)}\right).
\vspace{-1mm}
\end{equation}} And $\widetilde{d}_{e}^{\,2}, v_{e}\in V_{i}^{2}$ (the updated degrees of nodes in $V_{i}^{2}$) is defined as 
{\small
\begin{equation}\label{added degrees 2}
\notag \widetilde{d}_{e}^{2}=\begin{cases}
d_{e}^{2}, v_{e}\in O( T^{1})\\
 \mathbb{E}_{i}(d_{e}^{2}), v_{e}\in V^{1}\backslash O( T^{1})\\
 m,  v_{e}\in V^{2}_{i}\backslash \left(V^{1} \bigcup O(T^{1})\right).
\end{cases}
\end{equation}}Based on the analysis above, we can inductively obtain the expected degrees until time $T^{r+1}$ in the $r$-th trial. Specially,  
{\small\begin{equation}\label{expected degree}
\mathbb{E}_{i}(d_{e}^{r+1})=\widetilde{d}_{e}^{r}\cdot \prod_{s=1}^{m[n_{i}(T^{r})-n_{i}(T^{r})]} \left(1+\frac{1}{\sum_{v_{j}\in V_{i}^{r}}\widetilde{d}_{j}^{r-1}+(2s-1)}\right).
\end{equation}}Accordingly, $\widetilde{d}_{e}^{r}  (v_{e}\in V_{i}^{r})$ is defined as 
{\small \begin{equation}\label{added degrees r}
\notag \widetilde{d}_{e}^{r-1}=\begin{cases}
d_{e}^{r}, v_{e}\in O( T^{r-1})\\
 \mathbb{E}_{i}(d_{e}^{r-1}), v_{e}\in V^{r-1}_{i}\backslash O( T^{r-1})\\
 m,  v_{e}\in V^{r}_{i}\backslash \left(V_{i}^{r-1} \bigcup O(T^{r-1})\right).
\end{cases}
\end{equation}}
 \section{ Proof for Lemma 3}\label{proof for Kalman Gain}

 Lemma 3. The Kalman Gain in the refinement of $w_{e}$ in the $r$-th trial is determined by
\begin{displaymath}
\mathbf{G}_{e,r}=\mathbf{\Sigma}_{e, r-1}\cdot \mathbf{Q}_{e,r}^{-1},
 \end{displaymath}
where $\mathbf{Q}_{e,r}=\mathbf{\Sigma}_{e, r}+1$ denotes variance of the activating result via $e$.
\begin{proof}
In the refinement of $w_{e}$ in $r$-th trial, the Kalman Gain $\mathbf{G}_{e,r}$ is determined by minimizing the mean square estimation error of $w_{e,r}$, i.e.,
\vspace{-1mm}
\begin{align}\label{kalman1}
&\mathbf{G}_{e,r} = \mathop{\arg\min}_{M\in \mathbb{R}}\mathbb{E}\left[ \left(\overline{w}_{e,r}'-w_{e,r} \right)^2\right], \\
&\textnormal{where} \quad  \overline{w}_{e,r}' =\overline{w}_{e,r-1}' +M\cdot \left(z_{e,r}-\overline{w}_{e,r-1}'\right).
\vspace{-1mm}
\end{align}Here, $\mathbb{R}$ represents the set of  real numbers. By minimizing  the objective function in Eqn. (\ref{kalman1}), the Kalman Gain in refinement is formulated as $\mathbf{G}_{e,r}=\mathbf{\Sigma}_{e, r-1}\cdot \mathbf{Q}_{e,r}^{-1}$, where $\mathbf{Q}_{e,r}$ denotes the variance of the activating result via edge $e$. With the consideration of both the variance of $w_{e}$ and the observing error, $\mathbf{Q}_{e,r}$ is formalized as $\notag \mathbf{Q}_{e,r}=\mathbf{\Sigma}_{e, r-1}+\sigma^{2}$,
where $\sigma^{2}$ denotes the square observing error of Bernoulli distribution $\mathcal{B}(w_{e,r})$ with $0\leq\sigma^{2}\leq1$, and we set $\sigma^{2}$ as its maximum value $1$ (e.g., $w_{e,r}=0$ and $z_{e,r}=1$). Then the distribution of the weight of edge $e$ is refined with $\mathbf{G}_{e, r}$ and $\mathbf{Q}_{e, r}$ as Eqn. (\ref{mean}) and Eqn. (\ref{square}).
\end{proof}

 \section{ Proofs for Lemmas 4  and 6}\label{hardnessproof2}

Lemma \ref{submodular}. The influence function in Eqn. (\ref{problem}) is monotonous and submodular.

\begin{proof}
{\bf Monotonicity.} We consider an instance of $\mathbb{G}^{r}$, i.e., $\overline{\mathbb{G}}^{r}$ where the state of each edge $e$ is determined by flipping a coin of bias $w_{e,r}$. If the coin of edge from node $u$ to node $v$ flips, user $u$ can successfully influence user $v$ after he has been influenced. Let $S_{1}$ and $S_{2}$ denote two seed sets with $S_{1}\subseteq S_{2}$, and $I(S_{1}, \overline{\mathbb{G}}^{r})$ and $I(S_{2}, \overline{\mathbb{G}}^{r})$ respectively denote the weights sum of nodes influenced by $S_{1}$ and $S_{2}$ over $\overline{\mathbb{G}}^{r}$. For any node $a$ influenced by $S_{1}$, there must be an active path from a node in $S_{2}$ to $a$ since $S_{1}\subseteq S_{2}$. Thus we have  $I(S_{1}, \overline{\mathbb{G}}^{r})\leq I(S_{2}, \overline{\mathbb{G}}^{r})$, which demonstrates the monotonicity of influence function in Eqn. (\ref{problem}). 

{\bf Submodularity.} Furthermore, let $S_{3}=S_{1} \cup x$, $S_{4}=S_{2} \cup x$, $I(S_{3}\backslash S_{1}, \overline{\mathbb{G}}^{r})=I(S_{3}, \overline{\mathbb{G}}^{r})-I(S_{1}, \overline{\mathbb{G}}^{r})$ and $I(S_{4}\backslash S_{2}, \overline{\mathbb{G}}^{r})=I(S_{4}, \overline{\mathbb{G}}^{r})-I(S_{2}, \overline{\mathbb{G}}^{r})$. For a node $a$ that can be influenced by $x$ while cannot be influenced by $S_{2}$ over $\overline{\mathbb{G}}^{r}$, there is an active path from $x$ to $a$ while no active path from $S_{2}$ to $a$. Since $S_{1}\subseteq S_{2}$ and $S_{3}=S_{1} \cup x$, the node $a$ cannot be influenced by $S_{1}$ while can be influenced by $S_{3}$ over  $\overline{\mathbb{G}}^{r}$. Thus we have $I(S_{3}\backslash S_{1}, \overline{\mathbb{G}}^{r})\geq I(S_{4}\backslash S_{2}, \overline{\mathbb{G}}^{r})$, which demonstrates the submodularity of influence function in Eqn. (\ref{problem}). 
\end{proof}

Lemma 6. If the Inequalities (\ref{R}) and  (\ref{R1}) hold, with at least $(1-1/2n^{l'})$ probability, we have $\mathbb{E}[I(S^{r}, \mathbb{G}^{r})]\geq (1-1/e-\varepsilon)\cdot OPT$. 

\begin{proof}
Let $S$ be a seed set with size $K$. We say $S$ is a bad seed set if $\mathbb{E}[I(S, \mathbb{G}^{r})]\leq (1-1/e-\varepsilon)\cdot OPT$. Since the number of bad sets is at most $\binom{n}{K}$, proving Lemma \ref{approximation ratio} is equivalent to proving that any bad set $S$ has a probability of at most $n^{l'}/\binom{n}{K}$ to be returned by the NodeSelection phase. If $S$ is returned, there must be $F_{\mathcal{R}}(S)\geq F_{\mathcal{R}}(S^{r})$.
Thus the probability of $S$ being returned by the NodeSelection phase is upper bounded by $Pr[F_{\mathcal{R}}(S)\geq F_{\mathcal{R}}(S^{r})]$. Then 
{\small \begin{align}
\notag &Pr[F_{\mathcal{R}}(S)\geq F_{\mathcal{R}}(S^{r})]\\
\notag =&Pr\left[\frac{n'}{\theta'}F_{\mathcal{R}}(S)\geq \frac{n'}{\theta'}F_{\mathcal{R}}(S^{r})\right]\\
\label{R3} =&Pr\left[\frac{n'}{\theta'}F_{\mathcal{R}}(S)-\mathbb{E}[I(S, \mathbb{G}^{r})]\geq \frac{n'}{\theta'}F_{\mathcal{R}}(S^{r})-\mathbb{E}[I(S, \mathbb{G}^{r})]\right]
\end{align}}From Inequality (\ref{R1}) and the property of bad sets, we have 
\begin{align}
\notag &\frac{n'}{\theta'}F_{\mathcal{R}}(S^{r})-\mathbb{E}[I(S, \mathbb{G}^{r})]\\
\notag \geq& (1-1/e)(1-\varepsilon_{1})\cdot OPT-(1-1/e-\varepsilon)\cdot OPT\\
\notag=& (\varepsilon-(1-1/e)\varepsilon_{1})\cdot OPT. 
\end{align}Let $\varepsilon_{2}=\varepsilon-(1-1/e)\varepsilon_{1}$ and $\mathbb{E}[I(S, \mathbb{G}^{r})]=n'p$, then Eqn. (\ref{R3}) becomes 
\begin{align}
\notag &Pr\left[\frac{n'}{\theta'}F_{\mathcal{R}}(S)-n'p\geq \varepsilon_{2}\cdot OPT\right]\\
\notag =&Pr\left[F_{\mathcal{R}}(S)-\theta'p\geq \frac{\varepsilon_{2}OPT}{n'p}\cdot \theta'p\right].
\end{align}Let $\xi=\frac{\varepsilon_{2}OPT}{n'p}$, by the Chernoff bound, we have

\begin{align}
\notag &Pr\left[F_{\mathcal{R}}(S)-\theta'p\geq \xi\cdot \theta'p\right]\\
\notag \leq &\exp \left(-\frac{\xi^{2}}{2+\xi} \cdot \theta'p\right)\\
\notag=&\exp\left(-\frac{\varepsilon_{2}^{2}\cdot OPT^{2}}{2n'^{2}p+\varepsilon_{2}\cdot OPT \cdot n'}\cdot \theta'\right)\\
\notag\leq&\exp\left(-\frac{\varepsilon_{2}^{2}\cdot OPT^{2}}{2n'(1-1/e-\varepsilon)\cdot OPT+\varepsilon_{2}\cdot OPT \cdot n'}\cdot \theta'\right)\\
\notag\leq&\exp\left(-\frac{(\varepsilon-(1-1/e)\cdot\varepsilon_{1})^{2}\cdot OPT}{(2-2/e)\cdot n'}\cdot \theta'\right)\\
\notag \leq &\exp\left(-\log(\binom{n}{K}\cdot (2n^{l'}))\right)\\
\notag \leq &n^{-l'}/\left[2\cdot \binom{n}{K}\right].
\end{align}Then by the union bound, the probability that NodeSelection phase returns a bad seed set is upper bounded by $\left[n^{-l'}/\left(2\cdot \binom{n}{K}\right)\right]\cdot \binom{n}{K}=1/2n^{l'}$. Thus, with a probability of at least $(1-1/2n^{l'})$, the NodeSelection phase returns a seed set that satisfies $\mathbb{E}[I(S^{r}, \mathbb{G}^{r})]\geq (1-1/e-\varepsilon)\cdot OPT$. 
\end{proof}

\section{Proof for Lemma 7}\label{proof for efficiency}

Lemma 7. The time complexity of {\bfseries Evo-IMM} is $O\big((K+l\ \left((n+m)+\frac{n}{OPT}\big)\log n/\varepsilon^{2}\right)$, where $n=|\mathbb{V}^{r}|$ and $m=|\mathbb{E}^{r}|$.

     \begin{proof}
     We divide the time costs of {\bf Evo-IMM} into two parts, Sampling \& NodeSelection and Priority-based sampling. 
     
{\bf Sampling \& NodeSelection.} Let $EPT$ denote the expected number of edges pointing to the nodes in an ERR-set. Since generating an ERR-set needs to traverse all the edges inside the set , each ERR-set costs $O(EPT)$ time in the sampling phase. According to the analysis in the general IMM framework \cite{IMM}, the size of $\mathcal{R}$ is $|\mathcal{R}|=O\big((K+l)n\log n \cdot  \varepsilon^{-2}/OPT\big)$. And the NodeSelection runs in the time linear to the size of $\mathcal{R}$ since it corresponds to the standard greedy approach for the maximum coverage problem. Thus Sampling and NodeSelection cost $O(|R|\cdot EPT)$ expected time. 
Besides, by the analysis of the general IMM framework, we have $n\cdot EPT\leq m\cdot OPT$. Then the time complexity becomes
\begin{displaymath}
O(|\mathcal{R}|\cdot EPT)=O\left((K+l)(n+m)\log n \cdot  \varepsilon^{-2}\right). 
\end{displaymath}

{\bf Priority-based sampling.} Priority-based sampling is called by Sampling to generate ERR-sets. In each calling, Priority-based sampling deletes the root node of the ERR-set that has just been sampled to update the sampling interval and then generates a new ERR-set. Thus the total time cost of Priority-based sampling is linear to $O(|\mathcal{R}|)$, which is the number of ERR-sets needed.
 
Summing up the time costs of the two parts, the expected total time involved in the evolving seed selection algorithm {\bfseries Evo-IMM}  is $O\big((K+l)\big((n+m)+\frac{n}{OPT}\big)\log n/\varepsilon^{2}\big)$. 
\vspace{-2mm}
\end{proof}

 \section{ Proof for \textnormal{ Inequality (22)}}\label{AC}
 For a node $u\in T_{r,v}$, we define the probability that it is being influenced as $h_v(u,\vec{w})$. Thus, if $u\in S^{r},$ $h_v(u,\vec{w})=1$, otherwise we have 
  \begin{displaymath}
h_v(u,\vec{w}) = 1 - \prod_{u'\in\mathcal{C}(u) } \Big( 1- w_{u'u} h_v(u',\vec{w})\Big), 
	\end{displaymath} where $\mathcal{C}(u)$ denotes the set of neighbors of node $u$, and  $w_{u'u}$ denotes the weight of edge between $u'$ and $u$. 
Considering the difference between influence diffusion under  $\vec{w}_{r}$ and  $\vec{w}_{r}'$, we have

{\small \begin{align}
\notag &h_v(u,\vec{w}_{r}') -h_v(u,\vec{w}_{r}) \\
\notag = &\prod_{u'\in\mathcal{C}(u) } \Big( 1- w_{u'u}^{'} h_v(u',\vec{w}_{r}')\Big)-\prod_{u'\in\mathcal{C}(u) } \Big( 1- w_{u'u} h_v(u',\vec{w}_{r})\Big)\\
\label{a7} \leq  &\sum_{u'\in\mathcal{C}(u) } \Big(w_{u'u}^{'}h_v(u',\vec{w}_{r}')-w_{u'u}h_v(u',\vec{w}_{r})\Big)\\ \notag \leq & \sum_{u'\in\mathcal{C}(u) } \Big(h_v(u',\vec{w}_{r}') - h_v(u',\vec{w}_{r}) + (w_{u'u}^{'}-w_{u'u})h_v(u',w_{u'u})\Big)\\
\label{41}  \leq & \sum_{u'\in\mathcal{C}(u) } \Big ((w_{u'u}^{'}-w_{u'u})h_v(u',w_{u'u})\Big)
\end{align}}The Inequality (\ref{a7}) is obtained according to the following lemma: 
\begin{lemma}\label{a8} 
(\cite{semi-bandit}.) Given $a_1,\cdots,a_n,b_1,\cdots,b_n \in (0,1)$ and $a_{k} \leq b_{k}, k = 1,\cdots,n$, then 
	\begin{displaymath}
			\prod_{k=1}^n b_k - \prod_{k=1}^n a_k \leq \sum_{k=1}^{n}(b_k - a_k). 
		\end{displaymath}
\end{lemma}On the other hand, if node  $u\in T_{r,v}$ is observed, there must be a fact that the edge from $S^{r}$ to $u$ is triggered. Thus for the root node $v$ of $T_{r,v}$, based on Inequality (\ref{41}), we further have the following conclusion from the edge level, i.e., 
	\begin{equation}\label{a9}
		h_v(v,\vec{w}_{r}') -h_v(v,\vec{w}_{r}) \leq \sum_{w_{e}\in T_{r,v}}  \mathbb{E}\left[\mathbb{I}(o_e^r)(w_{e, r}'- w_{e, r})\right].
	\end{equation}Since $I(S^{r}, \vec{w}_{r}')$ and $I(S^{r}, \vec{w}_{r})$ denote the expected influenced size of seed set $S^{r}$ under $\vec{w}_{r}'$ and $\vec{w}_{r}$ respectively, by summing Eqn. (\ref{a9}) over $V^{r+1}\setminus S^{r}$, we  have
 \begin{equation}
\notag \mathbb{E}\left[I(S^{r}, \vec{w}_{r}')-I(S^{r},\vec{w}_{r})\right] \leq  \sum_{v\in V^{r+1}\setminus S^{r}}\sum_{v_{e}\in T_{r,v}}\mathbb{E}\left[\mathbb{I}(o_e^r)(w_{e,r}' -w_{e,r})\right]. 
\end{equation}Thus we end the proof for Inequality (\ref{L1}). 
\begin{figure*}[t]
 \centering
 \vspace{-3mm}
\centering
  \includegraphics[width=1\textwidth]{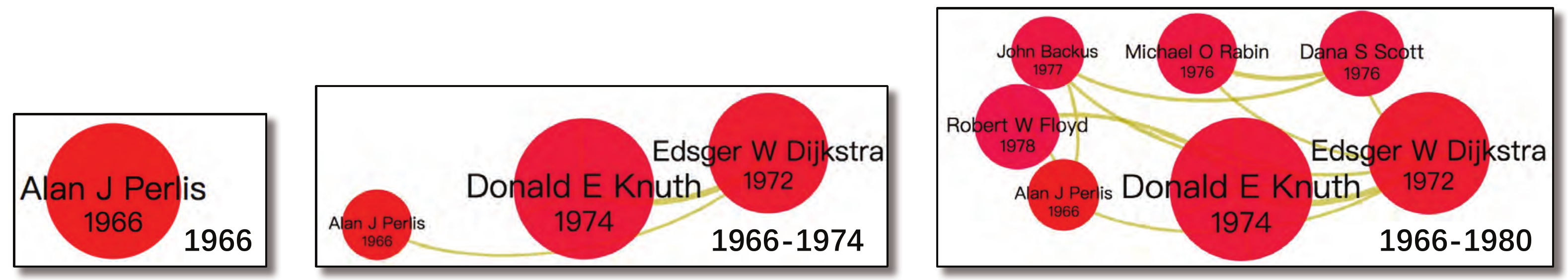}
  \vspace{-2mm}
  \caption{The evolving cooperation network of Turing awardees}\label{Turing4}
   \vspace{-1mm}
  \end{figure*}
    \begin{figure*}[t]
 \centering
 \vspace{-3mm}
\centering
  \includegraphics[width=0.9\textwidth]{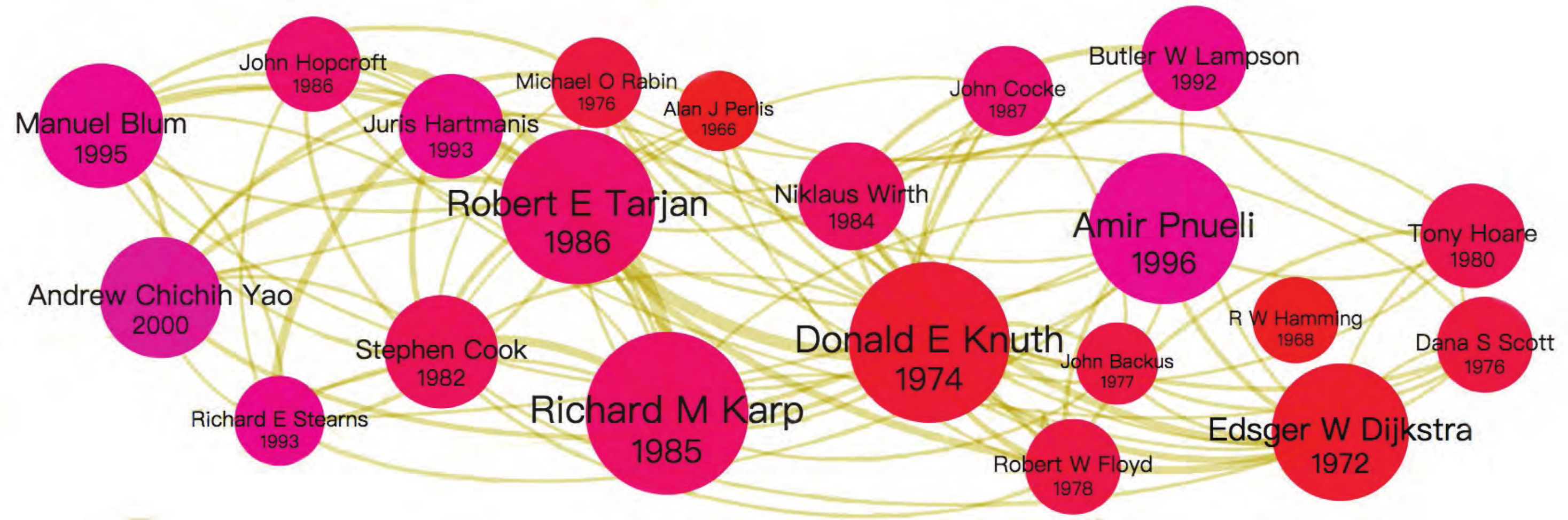}
  \vspace{-2mm}
  \caption{The evolving cooperation network of Turing awardees (1966-2000)}\label{Turing5}
   \vspace{-4mm}
  \end{figure*}
 \section{ Proofs for Lemmas 10, 11 and 12}
\begin{lemma}\label{b1}
If $0<k\leq 1$ and $r\leq 2$, then 
\begin{equation}
\frac{3}{r^{k}} - \frac{3}{(r-1)^k + 3} \geq \frac{1}{r^{2k}}.
\end{equation}
\end{lemma}
\begin{proof}
\begin{displaymath}
\frac{3}{r^{k}} - \frac{3}{(r-1)^k + 3} = \frac{9 + 3(r-1)^k - 3r^{k}}{r^{k}\big((r-1)^k + 3\big)}.
\end{displaymath}For $r^k - (r-1)^{k}$, we have $r^k - (r-1)^{k}=(r^{k})'|_{\lambda}, \big ((r-1)\leq \lambda \leq r \big)$. Since $\frac{d^{2}(r^{k})}{d^{2}r}\leq 0$, then
\begin{displaymath}
(r^{k})'|_{\lambda}\leq (r^{k})'|_{r-1}=k(r-1)^{k-1}\leq k.
\end{displaymath}Since $r^k >1, r^k > (r-1)^k$, $r^{k}\big((r-1)^k + 3\big)\leq 4r^{2k}$. Thus 
\begin{displaymath}
 \frac{9 + 3(r-1)^k - 3r^{k}}{r^{k}\big((r-1)^k + 3\big)} \geq \frac{9 -3k }{4r^{2k}} \geq \frac{1}{r^{2k}}.
\end{displaymath}\end{proof}

\begin{lemma}\label{b2}
If $0<k<1, r\geq 1$, then
\begin{equation}
\sum_{r=1}^{R} \frac{1}{r^k} \leq \frac{1}{1-k}R^{1-k}.
\end{equation}
\end{lemma}
\begin{proof}The power series expansion of $(r-1)^k$ is 
{\small \begin{align}
\notag &(r-1)^k\\
&= r^k -k r^{k-1} +\frac{k(k-1)}{2}r^{k-2} + \cdots + (-1)^n\frac{\prod_{i=0}^{n-1}(k-i)}{n!} r^{k-n} + \cdots .
\end{align}}Note that from the third term to the end in above power series are all negative, thus 
\begin{align}
\notag &(r-1)^k \leq  r^k - k r^{k-1}\\
\notag &\frac{r^{1-k}-(r-1)^{1-k}}{1-k} \geq  \frac{1}{r^k}\quad (k \rightarrow (k-1)).
\end{align}Hence, 
\begin{displaymath}
\sum_{r=1}^{R} \frac{1}{r^k}\leq \frac{1}{1-k} (R^{1-k} -1) \leq \frac{1}{1-k}R^{1-k} .
\end{displaymath}
\end{proof}
\begin{lemma}\label{b3}
If $0<k\leq 1, r\geq 2, x>0$
\begin{equation}
\frac{3}{r^{k}}+\frac{x}{4} - \frac{x(r-1)^k+3}{(x+1)(r-1)^k + 3} \geq \frac{1}{r^{2k}}. 
\end{equation}
\end{lemma}
\begin{proof}
Similar to the proof of Lemma \ref{b1}, we have 
\begin{displaymath}
\begin{split}
&\frac{3}{r^k}+\frac{x}{4} - \frac{x(r-1)^k+3}{(x+1)(r-1)^k + 3}\\
=&\frac{x}{4} + \frac{3(x+1)(r-1)^k + 9 - 3r^k - xr^k(r-1)^k}{r^{k}\big((x+1)(r-1)^k + 3\big)} \\
\geq &\frac{x}{x+4} + \frac{3x-3k + 9  - xr^{2k}}{(x+4)r^{2k}}\\
\geq &\frac{3x+6}{(x+4)r^{2k}} \geq \frac{1}{r^{2k}} \quad \left(\frac{3x+6}{x+4} \geq \frac{3}{2}(x>0)\right)
\end{split}
\end{displaymath}Thus we complete the proof for Lemma \ref{b3}. 
\end{proof}

\begin{figure*}[t]
 \centering
 \vspace{-1mm}
\centering
  \includegraphics[width=0.9\textwidth]{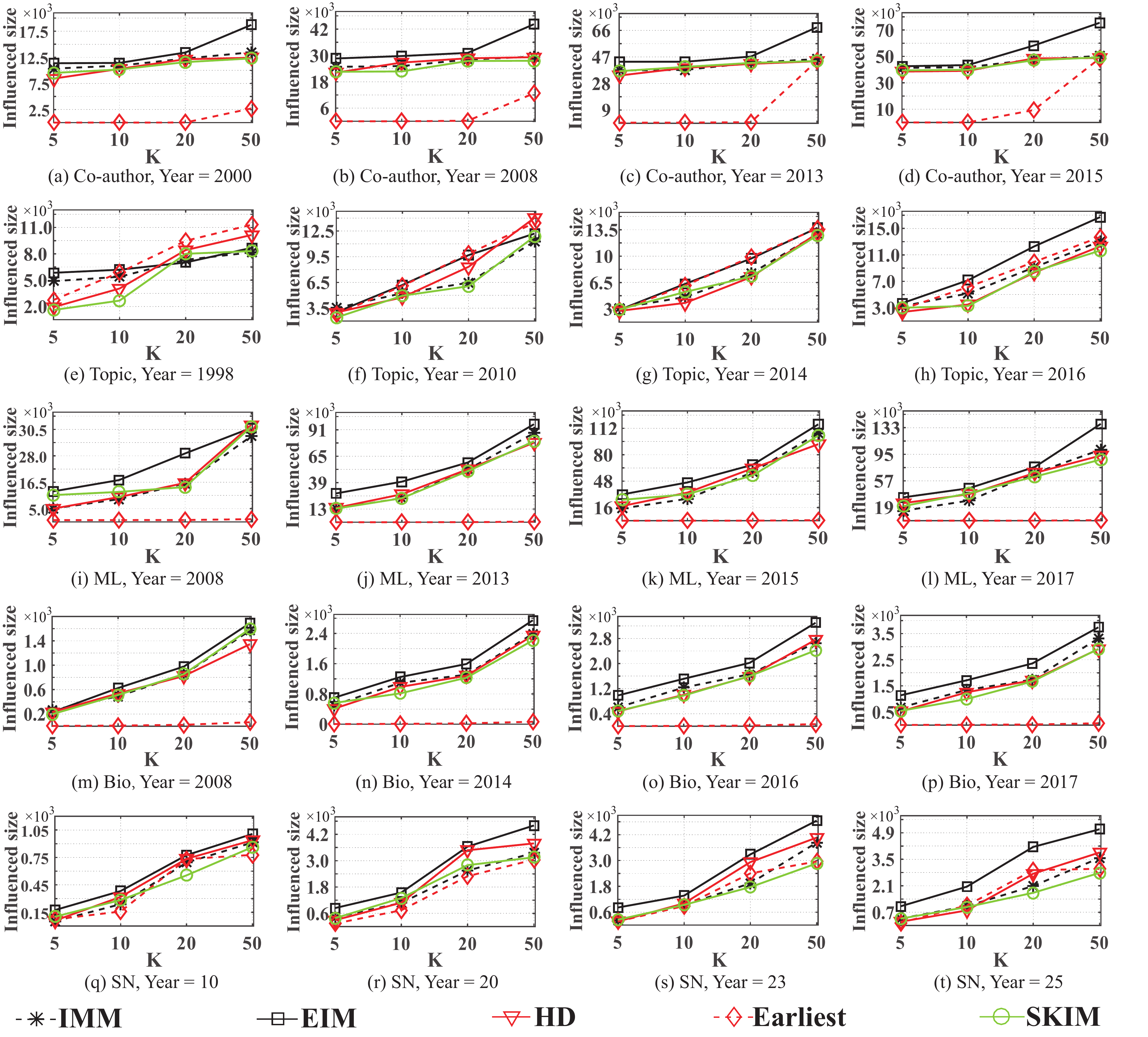}
  \vspace{-4mm}
  \caption{Influence size vs. K. }\label{KK}
   \vspace{-4mm}
  \end{figure*}
  \begin{figure*}[t]
 \centering
 \vspace{-1mm}
\centering
  \includegraphics[width=0.9\textwidth]{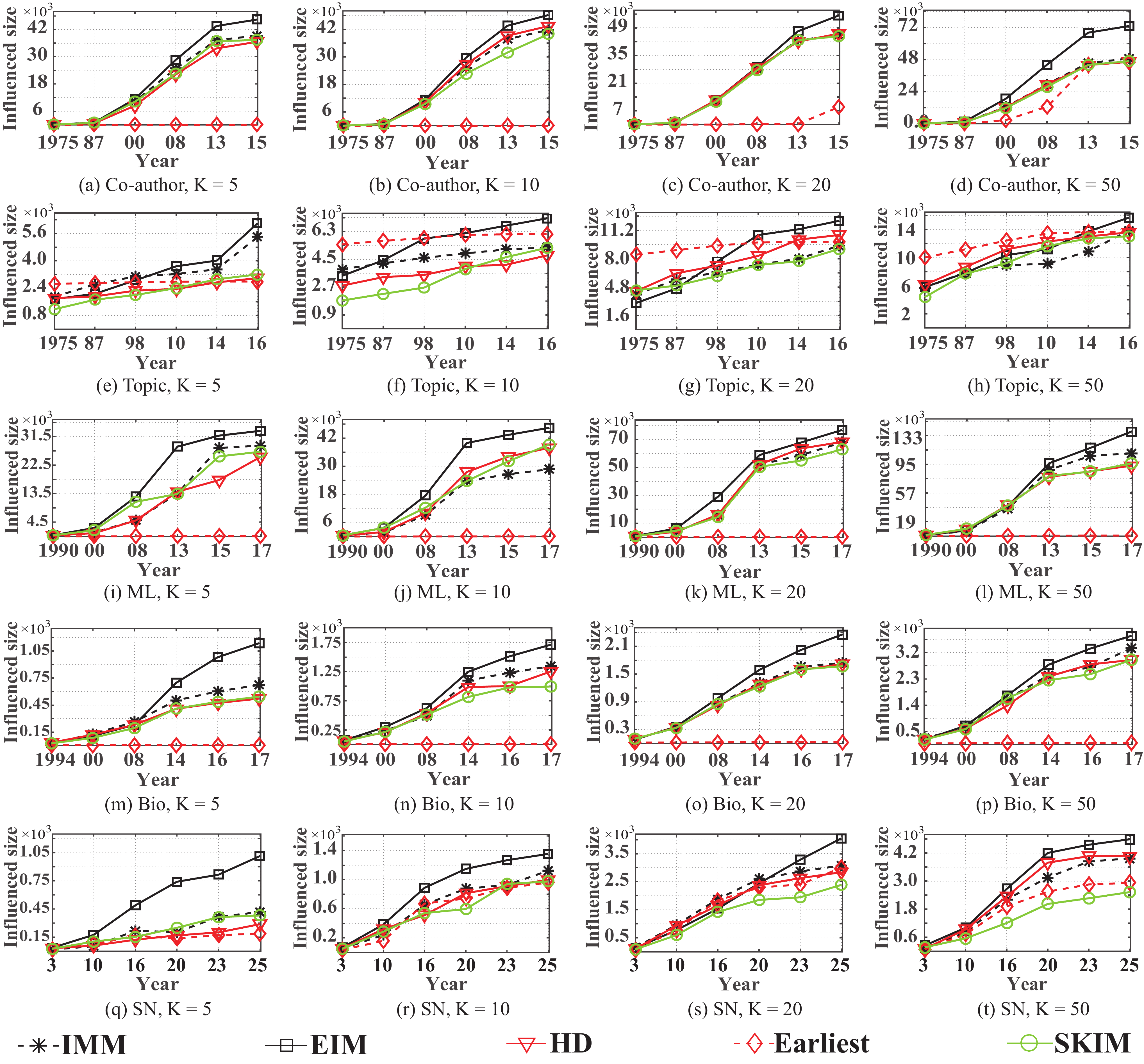}
  \vspace{-4mm}
  \caption{Influence size vs. Year. }\label{KL}
   \vspace{-4mm}
  \end{figure*}

\section{Supplementary Experimental Results}\label{Sexp}
 
\subsection{Evolving Network of Turing Awardees}
In Figures \ref{Turing4} and \ref{Turing5}, we provide additional interesting visualizations of the collaborative relationship among Turing Awardees (from 1966-2016) that we identify from the datasets of Coauthor, ML and Bio. The visualizations also serve as a
typical exmaple of evolving network with the joining time being the awarding
time of each awardee. The edges are based on both co-authorship and citations among the awardees. 
Such evolving cooperation network also validates heterogeneity in the attractiveness of different users (e.g.,  Donald E. Knuth has more new cooperators than others during $1974$-$1980$). 
\subsection{Complete Effectiveness Study}
Figures \ref{KK} and \ref{KL} presents the complete effectiveness study over the five evolving networks. We can see from both figures that $\mathbb{EIM}$ outperforms the four baselines owing to the network knowledge learning in each trial.
       \end{document}